\newif{\ifdraft}\drafttrue
\draftfalse

\ifdraft
\documentclass[draft,letter,envcountsame]{llncs}
\usepackage[notref]{showkeys}
\else
\documentclass[final,envcountsame]{llncs}
\fi 

\newif{\iflatin}\latintrue
\latinfalse
\usepackage[applemac]{inputenc}
\usepackage{times,amsmath,amssymb,tabularx,gastex,hhline,rotating,enumerate,xspace}
\usepackage[colorlinks,final]{hyperref}
\usepackage[usenames]{color}
\usepackage{showkeys}
\usepackage{tikz}
\usepackage{pgfplots}

\newcommand{\comment}[1]{}

\usepackage{prettyref}
\newcommand{\prref}[1]{\prettyref{#1}}
\iflatin
 \newrefformat{thm}{Thm.~\ref{#1}}
 \newrefformat{lem}{Lem.~\ref{#1}}
 \newrefformat{dfn}{Def.~\ref{#1}}
 \newrefformat{cor}{Cor.~\ref{#1}} 
 \newrefformat{prop}{Prop.~\ref{#1}}
 \newrefformat{sec}{Sect.~\ref{#1}}
 \newrefformat{kap}{Chap.~\ref{#1}}
 \newrefformat{ex}{Ex.~\ref{#1}}
 \newrefformat{app}{App.~\ref{#1}}
 \newrefformat{alg}{Alg.~\ref{#1}}
 \newrefformat{eq}{(\ref{#1})}%
 \newrefformat{tab}{Tab.~\ref{#1}}
 \newrefformat{fig}{Fig.~\ref{#1}}
 \newrefformat{rem}{Rem.~\ref{#1}}

\else
\newrefformat{thm}{Theorem~\ref{#1}}
\newrefformat{lem}{Lemma~\ref{#1}}
\newrefformat{dfn}{Definition~\ref{#1}}
\newrefformat{cor}{Corollary~\ref{#1}}
\newrefformat{prop}{Proposition~\ref{#1}}
\newrefformat{sec}{Section~\ref{#1}}
\newrefformat{kap}{Chapter~\ref{#1}}
\newrefformat{ex}{Example~\ref{#1}}
\newrefformat{rem}{Remark~\ref{#1}}
\newrefformat{fig}{Figure~\ref{#1}}
\newrefformat{app}{Appendix~\ref{#1}}
\newrefformat{eq}{(\ref{#1})}
\newrefformat{tab}{Table~\ref{#1}}
\fi

\newcommand{\Prob}[1]{\mathrm{Pr}\hspace*{-0.1pt}\left[ \mathinner{#1}\right]}

\newcommand{\simG}{\sim_{\BG}}
\newcommand{\lk}{\mathbf{\lfloor}}
\newcommand{\rk}{\mathbf{\rceil}}

\newcommand{\set}[2]{\left\{#1\mathrel{\left|\vphantom{#1}\vphantom{#2}\right.}#2\right\}}

\newcommand{\oneset}[1]{\left\{\mathinner{#1}\right\}}
\newcommand{\os}{\oneset}

\newcommand{\smallset}[1]{\left\{\mathinner{#1}\right\}}

\newcommand{\abs}[1]{\left|\mathinner{#1}\right|}
\newcommand{\Abs}[1]{\left\Vert\mathinner{#1}\right\Vert}

\newcommand{\floor}[1]{\left\lfloor\mathinner{#1} \right\rfloor}
\newcommand{\ceil}[1]{\left\lceil\mathinner{#1} \right\rceil}

\newcommand{\gen}[1]{\left< \mathinner{#1} \right>}
\newcommand{\Gen}[2]{\left< \mathinner{#1} \mid \mathinner{#2}\right>}

\newcommand{\scalp}[2]{\langle {#1}\,,\,{#2}  \rangle}
\newcommand{\scalb}[2]{\langle  {#1}\,,\,{#2}  \rangle_{\! \bet}}

\newcommand{\N}{\ensuremath{\mathbb{N}}}
\newcommand{\Z}{\ensuremath{\mathbb{Z}}}
\newcommand{\Q}{\ensuremath{\mathbb{Q}}}
\newcommand{\R}{\ensuremath{\mathbb{R}}}

\newcommand{\TC}{\ensuremath{\mathsf{TC}^0}}
\newcommand{\NC}{\ensuremath{\mathsf{NC}}}

\newcommand{\sdZ}{\ensuremath{\Z[1/2] \rtimes \Z}}

\renewcommand{\phi}{\varphi}
\newcommand{\eps}{\varepsilon}
\newcommand{\e}{\eps} 

\newcommand{\alp}{\alpha}
\newcommand{\bet}{\beta}
\newcommand{\gam}{\gamma}
\newcommand{\del}{\delta}

\newcommand{\Sig}{\Sigma}

\newcommand{\Del}{\Delta}

\newcommand\DD{\Delta}
\newcommand\GG{\Gamma}
\newcommand\LL{\Lambda}
\newcommand\OO{\Omega}

\newcommand{\Oh}{\mathcal{O}}

\newcommand{\cA}{\mathcal{A}}

\newcommand{\Breduced}{Britton-re\-du\-ced\xspace}

\newcommand{\Breduction}{Britton re\-duc\-tion\xspace}

\newcommand{\WP}{word problem\xspace}
\newcommand{\CP}{conjugacy problem\xspace}

\newcommand{\BS}[2]{\ensuremath{\mathrm{\bf{BS}}_{#1,#2}}}
\newcommand{\BG}{\ensuremath{\mathrm{\bf{G}}_{1,2}} 
}

\newcommand{\smalloverline}[1]
{{\mspace{1mu}\overline{\mspace{-1mu}#1\mspace{-1mu}}\mspace{1mu}}}
\newcommand{\ov}[1]{\smalloverline{#1}}
\newcommand{\oi}[1]{{#1}^{-1}}

\newcommand{\wt}[1]{\widetilde{#1}}
\newcommand{\wh}[1]{\widehat{#1}}

\newcommand\tow{\mathop \tau} 

\newcommand{\proba}{probability\xspace}
\newcommand{\tr}{triple-representation\xspace}
\newcommand{\PC}{power circuit\xspace}

\newcommand{\IFF}{if and only if\xspace}

\newcommand\lds{,\ldots ,} 
\newcommand{\sse}{\subseteq}
\newcommand{\es}{\emptyset}
\newcommand{\sm}{\setminus}
\newcommand{\OS}{\oneset{-1,0,+1}}

\newcommand\ei[1]{{\emph{#1}\xspace}\index{#1}}

\newcommand\ie{i.e., }
\newcommand\eg{e.g.\xspace}
\begin{document}
\title{Conjugacy in Baumslag's group, generic case complexity, and division in power circuits}
\author{Volker Diekert\inst{1}, Alexei G.{} Myasnikov\inst{2}, Armin Wei{\ss}\inst{1}}
\institute{FMI, Universit\"at Stuttgart,
Universit\"atsstr. 38, D-70569 Stuttgart, Germany\and 
Department of Mathematics, Stevens Institute of Technology, Hoboken, NJ, USA}

\date{\today}
\maketitle
\begin{abstract}
The conjugacy is the following question in  algorithmic group theory:  given two words $x$, $y$ over  generators of a fixed group $G$, decide whether $x$ and $y$ are conjugated, \ie whether there exists some $z$ such that $zx z^{-1} =y$ in $G$.
 The conjugacy problem is more difficult
than the word problem, in general. 
We investigate the conjugacy problem for two prominent groups: the Baumslag-Solitar group $\BS12$ and 
the Baumslag(-Gersten) group $\BG$. 
The conjugacy problem in $\BS12$ is $\TC$-complete. To the best of our knowledge 
\BS12 is the first natural infinite non-commutative group where such a precise and low complexity is shown. The Baumslag group $\BG$
is an HNN extension of $\BS12$. We show that the 
conjugacy problem is decidable (which has been known before); but our results go far beyond decidability.  In particular, we are able to show that
conjugacy in $\BG$ can be solved in polynomial time in a strongly generic setting. 
This means 
that essentially  for  all  inputs conjugacy in $\BG$ can be decided efficiently. In contrast, we show that under a plausible assumption the average case complexity of the same problem is non-elementary. Moreover, we provide a lower bound for the conjugacy problem in $\BG$ by reducing the 
division problem in \PC{}s to the conjugacy problem in $\BG$. The complexity of the 
division problem in \PC{}s is an open and interesting problem in integer arithmetic. To date it is believed that this problem has non-elementary time complexity.

Another  contribution of the paper concerns a general statement
about HNN extension of the form $G = \Gen{H,b}{bab^{-1}=\phi(a), a \in A}$ with 
a finitely generated base group $H$. We show that the complement of 
$H$ is strongly generic \IFF $A \neq H \neq B$. This is the 
situation for $\BG$; and yields an important piece of information why it is possible to solve  conjugacy for $\BG$ in strongly generic polynomial time. Note also that the complement of 
$H$ is strongly generic \IFF the Schreier graph of $G$ with respect to the subgroup $H$  is non-amenable.
\end{abstract}
\section*{Introduction}\label{sec:intro} 
More than 100 years ago Max Dehn introduced the word problem and the conjugacy problem as fundamental decision problems 
in group theory.  Let $G$ be a finitely generated group. 
 \emph{Word problem}:  Given two words $x$, $y$ written in generators, decide whether $x = y$ in $G$. 
\emph{Conjugacy problem}: Given two words $x$, $y$ written in generators, decide whether $x \sim_G y$ in $G$, \ie decide whether there exists $z$ such that $zx z^{-1}= y$ in $G$.
In recent years, conjugacy played an important role in non-commutative cryptography, see \eg~\cite{CravenJ12,GrigorievS09,SZ1}. These applications use that is is easy to create elements which are conjugated, but to check whether two given elements are conjugated might be difficult even if the word problem is easy. In fact, there are groups where the word problem is easy but the conjugacy problem is undecidable \cite{Miller1}.  Frequently, in cryptographic applications the ambient group is fixed. 
 The focus in this paper is on the conjugacy problem in $\BG$. 
In 1969 Gilbert Baumslag defined the group $\BG$ as an example of a  one-relator group 
which enjoys certain remarkable properties. It was introduced as  an infinite non-cyclic group 
all of whose finite quotients are cyclic \cite{baumslag69}. In particular, it is not residually finite; but being one-relator it has a decidable word problem \cite{mag32}.  The group $\BG$ is generated by generators $a$ and $b$ subject to 
a single relation $bab^{-1} a = a^2bab^{-1}$. Another way to understand $\BG$ is to view it as an HNN extension of the even more prominent 
Baumslag-Solitar group $\BS{1}{2}$.
The group $\BS{1}{2}$ is defined by a single relation 
$tat^{-1} = a^2$ where $a$ and $t$ are generators\footnote{Adding a generator $b$ and a relation 
 $bab^{-1} = t$ results in $\BG$. Indeed, due to $bab^{-1} = t$, we can remove $t$ and we obtain exactly the presentation of $\BG$ above.}. The complexity 
of the word problem and conjugacy problem  in $\BS{1}{2}$ are very low; indeed, we show that they are $\TC$-complete. However, such a low complexity does not transfer to 
the complexity of the corresponding problems in HHN-extensions like $\BG$. Gersten showed that the Dehn function of $\BG$ is non-elementary \cite{gersten91}. 
Moreover, 
Magnus' break-down procedure \cite{LS01}
on $\BG$ is non-elementary, too. 
 This means that the time complexity for the standard 
 algorithm to solve the word problem in $\BG$  cannot be bounded by any fixed tower of exponentials. 
Therefore, for many years, $\BG$ was the simplest candidate for a group with an extremely difficult word problem.
However,  Myasnikov, Ushakov, and Won showed in \cite{muw11bg}  that the word problem of the Baumslag group is solvable in polynomial time! 
In order to achieve a polynomial time bound they introduced a versatile data structure for integer arithmetic which they called \emph{power circuit}. The data structure supports
+, $-$, $\leq$, and $(x,y)\mapsto 2^xy$, a restricted version of multiplication which includes exponentiation $x \mapsto 2^x$. 
Thus, by iteration it is possible to represent huge values (involving the tower function) by very small circuits. Still, all operations above can be performed in polynomial time.
On the other hand there are notoriously difficult arithmetical problems in \PC{}s, too. A very important one is division. The input are \PC{}s $C$ and $C'$ representing 
integers $m$ and $m'$; the question is whether $m$ divides $m'$. 
The problem is clearly decidable by converting $m$ and $m'$ into binary; but this procedure is non-elementary. So far, no idea for any better algorithm is known. It is plausible to assume that the problem ``division in \PC{}s''
has  no elementary time complexity at all. 

In the present paper we show a tight relation between the problems ``division in \PC{}s'' and conjugacy in $\BG$. 
Our results concerning the Baumslag-Solitar group $\BS12$, the Baumslag group $\BG$, its generic case complexity, and division in power circuits are as follows. 
\begin{itemize}
\item The conjugacy problem of $\BS12$ is $\TC$-complete. 
\item There is a strongly generic polynomial time algorithm for the conjugacy problem in $\BG$. This means, the difficult instances for the algorithm are  exponentially sparse, and therefore, on random inputs, conjugacy can be solved efficiently. 
\item If ``division in \PC{}s'' is non-elementary in the worst case, then the
conjugacy problem in $\BG$ is non-elementary on the average.
\item A random walk in the Cayley graph of $\BG$ ends with  exponentially decreasing \proba  in $\BS12$. In other terms, the Schreier graph  of
$\BG$ with respect to $\BS12$  is non-amenable.
\end{itemize}
Decidability of the \CP in $\BG$ is not new, it was shown in \cite{beese12}\footnote{It is unknown whether the conjugacy problem in one-relator groups is decidable, in general.} and 
 decidability outside a so-called ``black hole'' follows already from \cite{BorovikMR07}.
 Our work improves Beese's work leading to a polynomial time algorithm outside a proper subset 
 of the ``black hole'' (and decidability everywhere). Thus, our result underlines that in special 
 cases like $\BG$ much better results than stated in \cite{BorovikMR07} are possible. Let us also note that there are  undecidable problems (hence no finite average case complexity is defined), like the 
halting problem for certain encodings of Turing machines, which have  generically linear time partial solutions. 
However, many of these examples depend on encodings and special purpose constructions.
In our case we consider a natural problem where the average case complexity is defined, but the only known algorithm to solve it runs in non-elementary time on the average. Nevertheless, there is a polynomial $p$ (roughly of degree $4$) 
such that the  probability that the same algorithm requires more than $p(n)$ steps on random inputs converges  exponentially fast to zero. The main technical difficulty in establishing a 
strongly generic polynomial time complexity is to show that a random 
walk of length $n$ in the  Cayley graph of $\BG$ ends with \proba less than 
$(1-\eps)^n$ 
in the subgroup $\BS12$ for some $\eps > 0$. Random walks in infinite graphs are widely studied in various areas, see \eg~\cite{woess94} or the textbook 
\cite{woess2000}. In \prref{sec:backbase} we prove a  general statement
about HNN extension of the form $G = \Gen{H,b}{bab^{-1}=\phi(a), a \in A}$ with 
a finitely generated base group $H$ and  $\Del$ a finite symmetric set of generators for $G$. 
We show that the complement of $H$ (inside $\Del^*$) 
 is strongly generic \IFF $A \neq H \neq B$. With other words, 
 the Schreier graph $\GG(G,H,\Del)$ is non-amenable if and only if $A\neq H \neq B$. (For a definition of amenability and its equivalent characterizations see \eg \cite{CeccheriniSilbersteinGH98,KMSS1}.) This  applies to $\BG$  because it is an HNN extension where $A \neq H \neq B$. 
However, in the special case of $\BG$ we can also apply a 
technique quite different from the general approach. In \prref{sec:dyckwordpair} we define a ``pairing'' between random walks in the Cayley graph and Dyck words. We exhibit an 
$\eps >0$ such that 
for each Dyck word $w$ of length $2n$  the \proba that a pairing with $w$ evaluates to $1$ is bounded by $(1/4-\eps)^{n}$. The result follows since there are at most 
$4^n$ Dyck words. 

\subsection*{Notation and preliminaries}\label{sec:notation}
\noindent{\bf Words.} An \emph{alphabet} is a (finite) set $\Sig$; an element $a \in \Sig$ is called a  \emph{letter}. The set $\Sig^n$
forms the set of \emph{words} of length $n$. The length of $w\in \Sig^n$
is denoted by $\abs w$. The set of all words is denoted by $\Sig^*$. It is the free monoid over $\Sig$. Let $a \in \Sig$ be a letter and  $w\in \Sig^*$. The number of occurrences of $a$ in $w$ is denoted by  ${\abs w}_a$. Clearly, $\abs w = \sum_{a \in \Sig}{\abs w}_a$. If we can write 
$w = uxv$, then we call $x$ a \emph{factor} of $w$; and we say that $w = uxv$ is a \emph{factorization}.

\noindent{\bf Functions.} We use standard $\Oh$-notation for functions from $\N$ to 
non-negative reals $\R^{\geq 0}$. (This includes of course $\OO$- and $\Theta$-notation.) The \ei{tower function} $\tow:\N \to \N$
is defined by $\tow(0) = 0$ and $\tow(i+1) = 2^{\tow(i)}$ for $i \geq 0$. 
It is primitive recursive. 
We say that a function $f:\N \to \R^{\geq 0}$ is \ei{elementary}, if the growth of $f$ can be bounded by a fixed number of exponentials. It is called \ei{non-elementary} if it is not elementary, but $f(n) \in \tow(\Oh(n))$. 
\iflatin \else Thus, in our paper non-elementary means a lower and an upper bound. \fi

\noindent{\bf Circuit complexity.} We deal with various complexity measures. 
On the lowest level we are interested in problems which can be decided 
by (uniform) $\TC$-circuits. These are circuits of polynomial size 
with constant depth where we allow  
Boolean gates and  
majority gates, which evaluate to $1$ \IFF the majority of inputs is $1$.
For a precise definition and 
uniformity conditions we refer to the textbook \cite{Vollmer99}. 
\iflatin \else
 $\TC$ circuits can be simulated by $\NC^1$ circuits, \ie circuits of 
logarithmic depth where only Boolean gates of constant fan-in are allowed. Thus, 
$\TC$ is a very low parallel complexity class. Still it is amazingly powerful with respect to arithmetic. In particular, we shall use Hesse's result that 
division of binary integers can be computed by a uniform family of 
\TC-circuits \cite{hesse01,HeAlBa02}.
\fi

\noindent{\bf Time complexity.} 
A uniform family of $\TC$-circuits computes a polynomial time computable function.
We use a standard notion for worst-case and for average case complexity
and random access machines (RAMs) as machine model.  An algorithm $\cA$  computes a function between domains $D$ and $D'$. In our applications $D$ comes always with a 
natural partition $D = \bigcup\set{D^{(n)}} {n\in \N}$ where each 
$D^{(n)}$ is finite. 
\iflatin \else 
The time complexity $t_\cA$ is defined by 
$t_\cA(n) = \max\set{t_\cA(w)}{w \in D^{(n)}}$.  Assuming a uniform distribution among elements in $D^{(n)}$, the average case complexity is defined by
$ {\mathrm{av}}_\cA(n) = \frac{1}{\vert{D^{(n)}}\vert}\sum_{w \in D^{(n)}}t_\cA(w).\label{eq:av}$ \fi

\noindent{\bf Generic case complexity.}
For many practical applications the ``generic-case behavior'' of an algorithm is more important 
than its average-case or worst-case behavior. We refer to \cite{KMSS1,KMSS2} where the foundations of this theory were developed  
and to \cite{MyasnikovSU08} for applications in cryptography. The notion of \emph{generic complexity} refers to partial algorithms which are defined on a  (strongly) generic set $I\sse D$.  Thus, they may refuse to give an answer outside 
$I$, but if they give an answer, the answer must always be correct. 
In our context it is enough to deal with totally defined 
algorithms and strongly generic sets. Thus, the answer is always computed and always correct, but the runtime is measured by a worst-case behavior over a strongly generic set $I\sse D$. 
Here  a set $I $ is called  \emph{strongly generic}, if 
there exists an $\eps > 0$ 
such that  $\abs {D^{(n)}\sm I } / \abs{D^{(n)}} \leq 2^{-\eps n}$ for almost all $n \in \N$. This means the probability to find a random string outside $I$ converges exponentially fast to zero. Thus, if an algorithm $\cA$ runs in polynomial time 
on a strongly generic set, then, for practical purposes, $\cA$  behaves as 
a polynomial time worst-case algorithm. This is true although the average time complexity of $\cA$ 
can be arbitrarily high. 

\noindent{\bf Group theory.} 
We use standard notation and facts from group theory as found in the 
classical text book \cite{LS01}. 
Groups $G$ are generated by some subset $S\sse G$. 
We let $\ov S = S^{-1}$ and we view $S \cup \ov S$ as an alphabet with 
involution; its elements are called \emph{letters}. We have $\ov {\ov a} = a$ for letters and also for words by letting $\ov{a_1 \cdots a_n} = \ov{a_n} \cdots \ov{a_1}$ where $a_i\in S \cup \ov S$ are letters. Thus, if $g\in G$ is given by a word $w$, then 
$\ov w = \oi g$ in the group $G$. For a word $w$ we denote by $\abs w$ its length.
We say that $w$ is \emph{reduced} if there is no factor $a\ov a$ for any letter.
It is called  \emph{cyclically reduced} if $ww$ is reduced. 
For words (or group elements) we write $x \sim_G y$ to denote conjugacy, \ie $x \sim_G y$ if and only if there exists some $z\in G$ such that $zx \ov z = y$ in $G$.
\iflatin \else For the decision problem ``conjugacy in $G$'' we assume that the input consists of cyclically reduced 
words $x$ and $y$ if not explicitly stated otherwise. \fi 
We apply the standard (so called
Magnus break-down) procedure for solving the word problem 
 in HNN extensions. Our calculations are fully 
 explicit and accessible with basic knowledge in 
 combinatorial group theory 
 
 \noindent{\bf Glossary.} $\TC$ circuit class. $x \sim_G y$ conjugacy in
 groups.  $(\GG,\del)$ \PC{}s. $\eps(P)$, $\eps(M)$ evaluation of nodes and markings. $\tow(n)$ tower function. Baumslag-Solitar group: $\BS12= \Gen{a,t}{tat^{-1} = a^2}$.  Baumslag group: 
 $\BG = \Gen{a,b}{bab^{-1}a = a^2ba^{-1}b^{-1}}$. Subgroup relations
 $A= \langle{a}\rangle$, $T= \langle{t}\rangle \leq \BS12 = \sdZ = H \leq \BG$. Standard symmetric set of generators for $\BG$ is $\Sig = \os{a,\ov a, b, \ov b}^*$ and $\ov z = \oi z$ in groups.
 
 \iflatin \noindent{\bf Proofs.} Missing proofs are in the appendix and in the full paper on the arxiv server.
 \else 
 \fi

\section{Power circuits}\label{PCs}
In binary a number is represented as a sum
$m= \sum_{i=0}^k b_i 2^i$ with $b_i \in \os {0,1}$. 
Allowing $b_i \in \os {-1,0,1}$ we obtain a ``compact representation'' of integers, which may require less non-zero $b_i$s than the normal representation.
The notion of \PC is due to \cite{MyasnikovUW12}. It generalizes 
compact representations and goes far beyond since it allows a compact representation of  tower functions. Formally: 
a \emph{\PC} of size $n$ is given 
by a pair $(\GG,\del)$. Here, 
$\GG$ is a set of $n$ vertices and $\del$ is a mapping $\del: \GG \times \GG\to \OS$.
The support of $\del$ is the subset $\DD \sse \GG \times \GG$ 
with $(P,Q) \in \DD \iff \del(P,Q) \neq 0$. 
Thus, $(\GG,\DD)$ is a directed graph. 
Throughout we require that $(\GG, \DD)$ is acyclic. 
In particular, $\del(P,P)=0$ for all vertices $P$. 
 A \ei{marking} is a mapping $M:\GG\to\OS$. We can also think of a marking as a
subset of $\GG$ where each element in $M$ has a sign ($+$ or $-$). 
If $M(P) =0$ for all 
$P\in \GG$ then we simply write $M= \es$. 
Each node $P\in \GG$ is associated in a natural way with a successor marking $\LL_P: \GG\to \OS, \; Q \mapsto \del(P,Q)$, consisting of the target nodes of outgoing arcs from $P$.
We define the \ei{evaluation} $\e(P)$ of a node ($\e(M)$ of a marking resp.)
bottom-up in the directed acyclic graph by induction:
\iflatin 
$\e(\es) = 0$,
$\e(P) = 2^{\e(\LL_P)}$ for a node $P$, and 
$\e(M) = \sum_{P}M(P)\e(P)$ for a marking  $M$.
\else
\begin{align*}
\e(\es) &= 0, \\
\e(P) &= 2^{\e(\LL_P)} &\text{for a node $P$}, \\
\e(M) &= \sum_{P}M(P)\e(P) &\text{for a marking  $M$}.
\end{align*}
\fi
Note that leaves evaluate to $1$, the evaluation of a marking is a real number, and the 
evaluation of a node $P$ is a positive real number. Thus, $\e(P)$ and $\e(M)$
are well-defined. We have $\eps(\LL_P) = \log_2(\eps(P))$, thus 
the successor marking plays the role of a logarithm. 
We are interested only in \PC{}s where all 
markings evaluate to integers; equivalently 
all nodes evaluate to some positive natural number in $2^{\N}$.

The \emph{\PC-representation} of an integer sequence $m_1, \ldots, m_k$ is given
by a tuple $(\GG, \del; M_1, \ldots, M_k)$
where $(\GG, \del)$ is a \PC and $M_1, \ldots, M_k$ are markings such that 
$\e(M_i) = m_i$. (Hence, a single \PC can store several different numbers;
a fact which has been  crucial in the proof of \prettyref{prop:WPgersten}, see \cite{dlu12higman}.)

\begin{example}\label{ex:binarybasis}
\iflatin We can represent every $n$-bit integer
as a \PC with $\Oh(n)$ vertices. 
\else
We can represent every integer in the range 
$[-n,n]$ as the evaluation of some marking in a \PC with node set $\oneset{P_{0,n} \lds P_\ell}$ such that $\e(P_i) =2^{i}$ for $0 \leq i \leq \ell$ and $\ell = \floor{\log_2 n}$.
Thus, we can convert the binary notation of an integer $n$ into a \PC
with $\Oh(\log\abs n)$ vertices and $\Oh((\log \abs n)\log\log \abs n)$ arcs.
\fi
\end{example}

\begin{example}\label{ex:powtow}
A \PC of size $n$ can realize $\tow(n)$ since a chain of $n$ nodes represents
$\tow(n)$ as the evaluation of the last node. 
\end{example}

\begin{proposition}[\cite{muw11bg,dlu12higman}]\label{prop:dlu_ijac}
The following operations can be performed in quadratic time. Input a 
\PC $(\GG,\del)$ of size $n$ and two markings $M_1$ and $M_2$. Decide whether $(\GG,\del)$ is indeed a \PC, \ie decide 
whether  all markings evaluate to integers. 
\iflatin 
If ``yes'':
Decide whether $\e(M_1) \leq \e(M_2)$; and 
compute a new \PC with markings $M$, $X$ and $U$ such that 
\begin{enumerate}
\item $\e(M) = \e(M_1) \pm \e(M_2)$.
\item $\e(M) = 2^{\e(M_1)} \cdot \e(M_2)$.
\item $\e(M_1) = 2^{\e(X)} \cdot  \e(U) $ and either $U = \es$ or $\e(U)$ is  odd.
\end{enumerate}
\else
If ``yes'':
\begin{itemize}
\item Decide whether $\e(M_1) \leq \e(M_2)$.
\item Compute a new \PC with markings $M$, $X$ and $U$ such that 
\begin{enumerate}
\item $\e(M) = \e(M_1) \pm \e(M_2)$.
\item $\e(M) = 2^{\e(M_1)} \cdot \e(M_2)$.
\item $\e(M_1) = 2^{\e(X)} \cdot  \e(U) $ and either $U = \es$ or $\e(U)$ is  odd.
\end{enumerate}
\end{itemize}
\fi
\end{proposition}
Let us mention that the complexity 
of the division  problem in \PC{}s is open. 
\iflatin The only known general algorithm transforms markings into binary numbers. This involves 
a non-elementary explosion.
\else
Here, the division problem is as follows. Given a
\PC of size $n$ and two markings $M_1$ and $M_2$,
decide whether $\e(M_1)\mid \e(M_2)$, \ie $\e(M_1)$ divides $\e(M_2)$.
We suspect that the division  problem in \PC{}s is extremely difficult. 
The only known general algorithm transforms $\e(M_1)$ and $\e(M_2)$
first in binary and solves division after that. So, the first part involves 
a non-elementary explosion. 
\fi

\section{Conjugacy in the Baumslag-Solitar group $\BS12$}\label{sec:cpbs}
The solution of the \CP in the Baumslag group $\BG$ relies on  the simpler solution
for the Baumslag-Solitar group  $\BS{1}{2}$. The aim of this section is to show that the \CP in $\BS{1}{2}$ is $\TC$-complete. The group
\BS12 is given by the presentation $\Gen{a,t}{ta\oi t = a^2}$. 
We have $ta = a^2 t$ and $a \oi t  = \oi t a^2$. This allows to represent all group elements by words of the form ${t}^{-p}a^r t^q$ with $p,q \in \N$ and $r \in \Z$. However, for $q\geq 0$,  transforming ${t}^{q}a^r$ into this form leads to $a^s t^q$ with $s = 2^q r$, so the word 
$a^s t^q$ is exponentially longer than the word ${t}^{q}a^r $.
We denote 
by $\Z[1/2] = \set{p/2^q\in \Q}{p,q \in \Z}$ the ring of 
\emph{dyadic fractions}. Multiplication by $2$ is an automorphism of the underlying additive group and therefore we can define the semi-direct product $\Z[1/2] \rtimes \Z$ as follows. Elements are pairs $(r,m) \in 
\Z[1/2] \times \Z$. The multiplication in $\Z[1/2] \rtimes \Z$ is
 defined by  $$(r,m)\cdot(s,q) = (r + 2^m s,m+q).$$ Inverses can be computed by the formula $(r,m)^{-1} = (-r \cdot 2^{-m}, -m)$.
 It is straightforward to show that 
 $a\mapsto (1,0)$ and $t\mapsto (0,1)$ defines an isomorphism
 between $\BS12$ and \sdZ. 
In the following we abbreviate  $\BS{1}{2}$ ($ = \sdZ$) by  $H$.
There are several options to represent a group element $g \in H$. 
In a \emph{unary representation} we write $g$ as a word over 
the alphabet with involution $\os{a, \ov a, t, \ov t}$.
Another way is to write  $g= (r,m)$ with $r \in \Z[1/2]$ and $m\in \Z$.
In the following we use both notations interchangeably.
 The \emph{binary representation} of  $(r,m)$ consists of $r$ written in binary (as floating point number) and  $m$ in unary. Let us write $(r,m)$ with $r = 2^k s$ and $k,s,m\in \Z$. We then have $(2^k s, m ) = (0,k)\cdot (s,m-k)$ and the  corresponding triple $[k,s,m-k]\in \Z^3$ is called the \emph{\tr} of $(r,m)$; it is not unique. 
The \emph{\PC  representation} of $g= [k,s,m-k]$ is given 
by a \PC and markings $K$, $S$, $L$ such that 
$\eps(K)= k$, $\eps(S)= s$, and $\eps(L)= m-k$.
Note that if $g \in \os{a, \ov a, t, \ov t}^n$ satisfies 
 $g =(r,m)\in H$, then  $\abs{r} \leq 2^n$ and $\abs m \leq n$. Thus, a transformation from unary to binary notation is on the safe side. 
 
\begin{proposition}\label{prop:multiplyBS}
Let $(r_1,m_1), \ldots , (r_n,m_n) \in \Z[1/2] \rtimes \Z$ given in binary representation for all $i$. Then there is a uniform construction of a $\TC$-circuit which calculates $(r,m)=(r_1,m_1)\cdots (r_n,m_n)$ in $\Z[1/2] \rtimes \Z$.
\end{proposition}

\newcommand{\proofmultiplyBS}{
\begin{proof}The statements concerning computations in $\TC$ are standard and can be found e.g.~in the textbook \cite{Vollmer99}. Let $N = \max\set{m_i, \floor{\abs{\log_2 r_i}}+1,n}{1\leq i \leq n}$. Since the $m_i$ are written in unary, we may assume for simplicity 
 that all $\abs{ m_i} \leq 1$ (hence, requiring $2$ bits) and all $r_i$ are written in binary using 
 exactly $2N$ bits ($N$ bits for the mantissa and $N$ for the exponent). Thus, we may assume that the input is a bit-string of length exactly $2(N^2+N)$. We have $m = \sum_{i=1}^n m_i$. By induction, using the equality $(r,m)(s,q) = (r+s\cdot 2^m,m+q)$, we see $r = \sum_{i=1}^n r_i\cdot 2^{k_i}$ where $k_i = \sum_{i=1}^{k-1} m_i$.
Since the numbers $k_i$ are bounded by $N$, they can be calculated by the iterated addition of the unary numbers $m_j$ for $j<i$, which is in $\TC$. In particular, $m$ can be calculated by a $\TC$-circuit.
The bit shift $r_i\mapsto r_i \cdot 2^{k_i}$ can be computed by a $\TC$-circuit.
It remains to calculate the iterated addition of binary numbers which is possible in $\TC$.
\qed\end{proof}
}
\iflatin \else \proofmultiplyBS \fi
The next proof uses a deep result of Hesse: integer division is in uniform $\TC$. 
\begin{proposition}\label{prop:conjBS}
Let $f=(r,m), g=(s,q)\in \Z[1/2] \rtimes \Z$ be given in binary representation. Then there is a uniform construction of a $\TC$-circuit which decides $f\sim_H g$.
\end{proposition}

\begin{proof}
Let $(r,m)\sim_H (s,q)$, \ie there are $k\in \Z$, $x\in \Z[1/2]$ with $(x,k) (r,m)= (s,q) (x,k)$. In particular,  $(r,m)\sim_H (s,q)$ if and only if $m=q$ and there are $k\in \Z$, $x\in \Z[1/2]$ such that 
\begin{equation}
 s =  r\cdot 2^k - x \cdot (2^m-1).\label{eq:conj}
\end{equation}
We have $(r,m)\sim_H (s,m)$
if and only if  $(-r,-m)\sim_H (-s,-m)$ since $(-p,-m)\sim_H (-p2^{-m},-m)
= (p,m)^{-1}$  for all $p\in \Z[1/2]$. Therefore, without restriction  $m \in \N$.
Since a conjugation with $t^k$ maps $(r,m)$ to $(2^kr,m)$, we may assume that $r,s \in \Z$ and $m \in \N$. For $m= 0$ this means $(r,0)\sim_H (s,0)$ \IFF there is some $k \in \Z$ such that $s =  r\cdot 2^k$. This can be decided in $\TC$. For $m=1$ we can choose $x = r-s$ and the answer is ``yes''. For $m \geq 2$ we can multiply \prettyref{eq:conj} by $2^\ell$ such that $x\cdot 2^\ell\in \Z$. We obtain $ 2^\ell\cdot(r \cdot 2^k - s) =  2^\ell x \cdot (2^m-1)$, \ie  $ 2^\ell\cdot(r \cdot 2^k - s) \equiv 0 \bmod (2^m-1)$.
The number $2$ is invertible modulo $2^m-1$ and its order is $m$.  Hence, actually for $m \geq 1$:
\begin{equation}
 (r,m)\sim_H (s,m) \iff \exists k \in \N: 0\leq k < m \wedge r \cdot 2^k - s\equiv 0 \bmod (2^m-1).\label{eq:frida}
\end{equation}
It can be checked 
whether such a $k$ exists using  Hesse's result for division \cite{hesse01,HeAlBa02}.
\qed\end{proof}


\begin{theorem}\label{thm:bstc}
 The word problem as well as the conjugacy problem in $\BS{1}{2}$ is $\TC$-complete.
\end{theorem}

\newcommand{\proofbstc}{
\begin{proof}
 By \prettyref{prop:multiplyBS} and \prettyref{prop:conjBS}, the conjugacy problem can be solved in $\TC$. The word problem is a special instance of the conjugacy problem and  the word problem in $\Z$ is $\TC$-hard in unary notation. This follows 
 because the $\TC$-hard problem \emph{MAJORITY} (see \cite{Vollmer99}) reduces uniformly to the unary word problem in $\Z$.
\qed\end{proof}
 }
 \iflatin \else \proofbstc\fi
\begin{remark}\label{rem:reddivBS}
Let us highlight that integer division can be reduced to the \CP in 
\BS12. For $m \geq 1$ we obtain as a special case of \prref{eq:frida} 
and a  well-known
fact from elementary number theory
\begin{equation}
 (0,m)\sim_H (2^s-1,m) \iff 2^m- 1\mid 2^s-1 \iff m \mid s.
 \label{eq:marie}
\end{equation}
If we allow a \PC representation for integers, then this reduction from division to conjugacy can be computed in polynomial time. Hence,  no elementary 
algorithm is known to solve the
\CP in $\BS12$ in \PC representation, whereas the word problem remains solvable in cubic time by \cite{dlu12higman}.
\end{remark}

\section{Conjugacy in the Baumslag group $\BG$}\label{sec:cpbg}
\label{sec:gersten}
The Baumslag group $\BG$ is an HNN extension of the Baumslag-Solitar group $\BS12$. We make this explicit. We let $\BS12$ be our base group,  generated by  $a$ and $t$. Again, $\BS12$ is abbreviated as $H$. The group $H$ contains infinite cyclic subgroups  $A = \gen a$ and $T = \gen t$
with  
$A \cap T = \os 1$. Let $b$ be a fresh letter which is 
added as a new generator together with the relation 
$bab^{-1} = t$. This defines the Baumslag group $\BG$. It is generated 
by $a,t,b$ with defining relations $tat^{-1} = a^2$ and $bab^{-1} = t$. 
However, the generator $t$ is now redundant and we obtain $\BG$
as a group generated by  $a,b$ with a single defining relation
$bab^{-1}a = a^2 bab^{-1}$. We represent elements of $\BG$
by $\bet$-factorizations. A \ei{$\bet$-factorization}
is written as a word
$z=\gam_{0}\bet_1\gam_{1}\ldots \bet_k\gam_{k}$  with $\bet_i \in \os{b, \ov b}$ and 
$\gam_i \in \os{a, \ov a, t, \ov t}^*$.
The number $k$ is called the 
\ei{$\bet$-length} and is denoted as $\abs{z}_\bet$ (i.e., 
$\abs{z}_\bet = \abs{z}_b + \abs{z}_{\ov b}$).
 A \emph{transposition} of a $\bet$-factorization $z = \gam_{0}\bet_1\gam_{1}\ldots \bet_k\gam_{k}$
is given as $z' = \bet_i\gam_{i}\ldots \bet_k\gam_{k}\gam_{0}\bet_1\gam_{1}\ldots \bet_{i-1}\gam_{i-1}$ for some $1 \leq i \leq k$. Clearly, 
$z\simG z'$ in this case. 
Throughout we identify a power ${c}^{-\ell}$ with ${\ov c}^\ell$ for letters $c$ and $\ell \in \N$.

\noindent{\bf  \Breduction{}s.}
A \ei{\Breduction} considers  some factor 
$\bet\gam\ov \bet$ with $\gam\in \os{a, \ov a, t, \ov t}^*$. There are two cases. First, if $\bet = b$ and 
$\gam =a^\ell$ in $H$ for some $\ell\in \Z$ then the factor $b\gam\ov b$ is replaced by $t^\ell$. Second, if $\bet = \ov b$ and  $\gam =t^\ell$ in $H$ 
for some $\ell\in \Z$ then the factor $\ov b \gam b $
is replaced by $a^\ell$. At most ${\abs z}_\bet$ \Breduction are possible
on a word $z$. Be aware! There can be a non-elementary blow-up in the exponents, see 
\prref{ex:blowup}.  If no \Breduction is possible, then the word 
$x$ is called \ei{\Breduced}. It is called \ei{cyclically \Breduced} 
if $xx$ is \Breduced. Britton reductions are effective because we can 
check whether $\gam =a^\ell$ (resp.\, $\gam =t^\ell$) in $H$. 
Thus, on input $x \in \os{a, \ov a, t, \ov t, b, \ov b}^*$ we can effectively calculate a \Breduced word $\wh x$ with 
$x = \wh x $ in $\BG$. The following assertions are standard facts for HNN extensions, see \cite{LS01}:
\begin{enumerate}
\item If $x$ is \Breduced then $x \in H$ \IFF $\abs{x}_\bet = 0$. 
\item If $x$ is \Breduced and $\abs{x}_\bet = 0$ then $x =1$ in $\BG$  \IFF $x =1$ in $H$.
\item Let $\bet_1\gam_{1}\ldots \bet_k\gam_{k}$ 
and $\bet'_1\gam'_{1}\ldots \bet'_k\gam'_{k}$ be $\bet$-factorizations of \Breduced words $x$ and $y$ such that $k \geq 2$  and $x = y$ in $\BG$.
Then we have  $k=k'$ and $(\bet_1,\ldots, \bet_k) = (\bet'_1,\ldots, \bet'_{k'})$. Moreover,
$\gam'_{1}\in \gam_{1}T$ if $\bet_2 = b$ and 
$\gam'_{1}\in \gam_{1}A$ if $\bet_2 = \ov b$.
\end{enumerate}

\begin{example}\label{ex:blowup}
Define  words 
$w_0=t$ and $w_{n+1} = b\,  w_{n}\,  a\, \ov{w_{n}}\,\ov b$
for $n \geq 0$. Then we have $\abs{w_{n}} = 2^{n+2} -3$ but 
$w_{n} = t^{\tow(n)}$ in $\BG$. 
\end{example}

The \emph{\PC-representation} of a $\bet$-factorization $\gam_{0}\bet_1\gam_{1}\ldots \bet_k\gam_{k}$
is the sequence
$(\bet_1 \lds  \bet_k)$ and a \PC $(\GG, \del)$ together 
with a sequence of markings $K_0, S_0, L_0  \lds K_k, S_k, L_k$ such that
$[\eps(K_i) , \eps(S_i), \eps(L_i)] = [k_i, s_i, \ell_i]$ is the triple representation of $\gam_i \in H$ for $1 \leq i \leq k$. 
It is known that the \WP of $\BG$ is decidable in cubic time. Actually a more precise statement holds. 
\begin{proposition}[\cite{muw11bg,dlu12higman}]\label{prop:WPgersten}
There is a cubic time algorithm which computes on input of a \PC representation of $x = \gam_{0}\bet_1\gam_{1}\ldots \bet_k\gam_{k}$
a \PC representation of a \Breduced word  (resp.\, cyclically \Breduced word) $\wh x$ such that $x = \wh x$ in $\BG$ (resp.\, $x \simG \wh x$). Moreover, the size for the \PC representation of $\wh x$ is linear in the 
size of the \PC representation of $x$.
\end{proposition}

\begin{remark}
 A polynomial time algorithm for the result in \prref{prop:WPgersten} has been given first in \cite{muw11bg}, 
 it has been  estimated by $\Oh(n^7)$. This was lowered in \cite{dlu12higman} to cubic time. 
\end{remark}

\begin{theorem}\label{thm:lea}
The following computation can be performed in time $\Oh(n^4)$.
Input: words $x,y \in \os{a,\ov a, 
b,\ov b}^*$. 
Decide whether $\abs{\wh x}_\bet >0$ for a cyclically \Breduced  form $\wh x$ of $x$. 
If ``yes'', decide $x \simG y$ and, in the positive case, compute 
a \PC representation of some $z$ such that $z x \ov z = y$ in $\BG$.
\end{theorem}

\newcommand{\inprooflea}{ 
 \noindent{\bf Case $n=1$.} 
 We have $x = \ov b (r,m)$ and $y = \ov b (s,q)$
 for some $(r,m), (s,q) \in \sdZ$.
  Now, $a^k x = y a^k$ in $\BG$ \IFF  $(0,k)(r,m) = (s,q) (k,0)$. 
  This forces $k = q-m$. Hence
    \begin{equation}\label{eq:futz} 
x \simG y \iff 2^{q-m} r = s + 2^q(q-m) 
\quad \text{for $n=1$.}
\end{equation}
\noindent{\bf Case $n \geq 2$ and $\eps_2 =+1$.} 
Then 
 $x= \ov b (r,m) b\gam_2\cdots b^{\eps_n}\gam_n$ and 
 $y= \ov b (s,q) b\gam'_2\cdots b^{\eps_n}\gam'_n$. 
 We have  $r \neq 0 \neq s$ since 
 $x$ and $y$ are  \Breduced. 
 For every $k \in \Z$ and every \Breduced  $\bet$-factorization 
 $\ov b  {\wt \gam}_1 b  \ldots  b^{\eps_n} {\wt \gam}_n$ for $a^k x{\ov a}^k$ we have
 ${\wt \gam}_1 \in t^{k} (r,m) T$, and hence ${\wt \gam}_1= (2^k r, p)$ for some $p \in \Z$. 
 We conclude that there is a unique $k \in \Z$ 
 such that $a^k x{\ov a}^k = \ov b \, (2^k r,p) b \cdots b^{\eps_n}{\wt \gam}_n \in \BG$, $p \in \Z$,  and $2^k r$ is an odd integer. 
 This means we may assume from the very beginning that $r$ and $s$ are odd integers. 
 Under this assumption, if  $a^k x a^{-k} = y$ in $\BG$ then  necessarily $k=0$ and hence $x = y$ in $\BG$. 
 We obtain the following algorithm to decide 
 $x \simG y$. 
 \begin{itemize}
\item For $\gam_1= (r,m)$ and $\gam'_1= (s,q)$
calculate unique $k,\ell \in \Z$ such that 
$2^k r$ and $2^{\ell}s $ are odd integers.
\item Decide whether $a^k x{\ov a}^k = a^\ell y{\ov a}^{\ell}  \in \BG$.
If ``yes'' then $x\simG y$ otherwise $x \not\simG y$.
\end{itemize}
{\bf Case $n \geq 2$ and $\eps_2 =-1$.} Then  
 $x= \ov b (r,m) \ov b\, \gam_2\cdots b^{\eps_n}\gam_n$ and 
 $y= \ov b (s,q) \ov b\, \gam'_2\cdots b^{\eps_n}\gam'_n$. 
 For every $k \in \Z$ we can write  
 $a^k x{\ov a}^k$ in some  \Breduced form which looks like  
 $\ov b \, {\wt \gam}_1 \ov b\cdots  b^{\eps_n}{\wt \gam}_n$.
Now,  ${\wt \gam}_1 \in t^{k} (r,m) A$. Thus, there is a unique $k \in \Z$ (necessarily 
 $k = -m$) such that ${\wt \gam}_1 = (p,0)$ for some $p \in \Z[1/2]$. Using the same arguments as above, 
 we obtain the following algorithm.
For $\gam_1= (r,m)$ and $\gam'_1= (s,q)$ decide whether ${a}^{-m} x a^m =  {a}^{-q} y a^q \in \BG$.
If ``yes'' then $x \simG y$ otherwise $x \not\simG y$.
}

\begin{proof}
 Due to \prref{prop:WPgersten},
we may assume that input words $x$ and $y$ are given as 
cyclically \Breduced words. In particular, $\wh x = x$ and 
$\abs{\wh x}_\bet =n > 0$. 
Let us write $x = \gam_0b^{\eps_1}\gam_1 \ldots b^{\eps_n}\gam_n$ as its  $\bet$-factorization
where $\eps_i = \pm 1$.  
If all $\eps_i =+1$ then we replace $x$ and $y$ by 
$\ov x$ and $\ov y$. Hence, without restriction there exists some $\eps_i =-1$. After a possible transposition we may assume that 
$x= b^{\eps_1} \gam_1 \cdots b^{\eps_n}\gam_n$ with $\eps_1 =-1$.  Since $y$ is cyclically \Breduced, too, Collins' Lemma (\cite[Thm.~IV.2.5]{LS01}) tells us several things: 
If $x\simG y$ then $\abs{y}_\bet =n$ and after some transposition 
the $\bet$-factorization of $y$ can be written as  $b^{\eps_1}\gam'_1 \cdots b^{\eps_n}\gam'_n$.  
 Moreover, still by Collins' Lemma, we now have 
 $x \simG y \iff \exists k \in \Z: y = a^k x a^{-k}$ in $\BG$.  The key  is  that  $k$ is unique and that we find an efficient way to calculate it\iflatin,\else.\fi\footnote{Beese calculates in \cite{beese12} this value $k$ and computes certain normal forms which are checked for equivalence. This leads  to an exponential time algorithm. } \iflatin see the appendix for the calculations. \fi

\iflatin\else \inprooflea \fi 

By \prref{prop:WPgersten}, the tests $a^k x{\ov a}^k =  y  \in \BG$ can be performed in cubic time. All other computations can be done in quadratic time by \prref{prop:multiplyBS}. Since all transpositions of the $\bet$-factorization for $y$ have to be considered this yields an $\Oh(n^4)$-algorithm. 
\qed\end{proof}

For the remainder of the section the situation is as follows: 
We have $x =(r,m) \in \sdZ$ and  $y =(s,q) \in \sdZ$, both can be assumed to be  in 
\PC representation. We may assume $x \neq 1 \neq y$ in $\BG$. After conjugation with some $t^k$ where $k$ is large enough we may assume that $r,m,s,q \in \Z$. 
If $m=0$ then we replace $x$ by $b x \ov b$. Hence, $m \neq 0$ and, by symmetry, $q \neq 0$, too. 
By \prref{eq:frida} and  ``division in \PC{}s'', we are able to 
 to test whether $(r,m)\sim_H (0,m)$ and $(s,q)\sim_H (0,q)$.
Assume that one of the answers is ``no''. Say, $(r,m)\not\sim_H (0,m)$.
Then there is no $h \in A \cup T \sse H$ such that 
$(r,m)\sim_H h$. Since then $\bet \gam (r,m) \ov \gam \ov \bet$ is \Breduced
for all $\bet\in \os{b, \ov b}$, $\gam\in \os{a, \ov a, t, \ov t}^*$
we obtain:
\begin{proposition}\label{prop:nixH}
Let $r,m \in \Z$, $m \neq 0$. If $(r,m)\not\sim_H (0,m)$ then $$(r,m)\simG (s,q) \iff (r,m)\sim_H (s,q).$$
\end{proposition}
By \prref{prop:nixH}, we may assume  $(r,m)\sim_H (0,m)$, $(s,q)\sim_H (0,q)$, and $(r,m)\not\sim_H (s,q)$. This involves perhaps non-elementary 
procedures. However, it remains to decide 
$(0,m)\simG(0,q)$, only.  The last test is polynomial time again, even for \PC{}s. 
\begin{proposition}\label{prop:verdi}
Let $m,q \in \Z$. Then we have 
$$(0,m)\simG (0,q) \iff (m,0)\sim_H (q,0) \iff \exists k \in \Z: m = 2^k q.$$
\end{proposition}

\newcommand{\proofverdi}{
\begin{proof} 
The assertion $(m,0)\sim_H (q,0) \iff \exists k \in \Z: m = 2^k q$ is clear since $(m,0)= a^m$ and $(q,0)= a^q$ in $H = \BS12$.
Let $(0,m)\simG (0,q)$. We have to show $(m,0)\sim_H (q,0)$
since the other direction is trivial. 
We have  $(q,0)\simG(0,q)$.
 Let 
$\gam_0b^{\eps_1}\gam_1 \cdots b^{\eps_n}\gam_n$ be a $\bet$-factorization of some $z$ 
with $n\in \N$ minimal such that $\ov z (q,0) z = (0,m)$. 
Since $\ov {\gam_0} (q,0) \gam_0 = (p,0)$ for some $p\neq 0$, 
we  have $n\geq 1$ and  $\eps_1 = -1$ because there has to occur a Britton reduction. Thus,  $b \ov {\gam_0} (q,0) \gam_0 \ov b= t^p$ in $\BG$. Now, $\ov {\gam_1} (0,p) \gam_1 \in A \cup T$
\IFF $\ov {\gam_1} (0,p) \gam_1 = (0,p)$. Thus, we may assume $\gam_1 = 1$ in $H$. Since $n$ is minimal we cannot have $\eps_2= +1$. Thus, we must have $n = 1$ and we may choose $z = \gam \ov b$ for some $\gam \in H$. 
This means $\ov z (q,0) z = b \ov \gam(q,0) \gam \ov b = (0,m)$
which implies $(m,0)\sim_H (q,0)$. 
\qed\end{proof}
}
\iflatin \else \proofverdi \fi

\begin{corollary}\label{thm:flo}
The following problem is decidable in at most non-elementary time.
Input: Power circuit representations $x,y$ for elements of $\BG$.
Question: $x \simG y$? 
\end{corollary}

\begin{corollary}\label{cor:divdif}
If there is no elementary algorithm to solve the division problem in \PC{}s  then the \CP  in the Baumslag group $\BG$ is non-elementary in the average case even for a unary representation of group elements. 
\end{corollary}

\begin{proof}
Assume that the \CP  in the Baumslag group $\BG$ is elementary
on the average. We give an  elementary algorithm to solve division in \PC{}s. 
 Let $(\GG, \del)$ be a \PC of size $n$ with markings $M$ and $S$ such that 
 $\eps(M) = m $ and $\eps(S) = s$. 
 For each node in $P \in \GG$ it is easy to construct a 
 word $w(P) \in \os{a,\ov a, b,\ov b}^*$ such that 
 $t^{\eps(P)}= w(P)$ in $\BG$ and $\abs{w(P)} \leq n^n$. Just follow the scheme from \prref{ex:blowup}. Hence,  in time $2^{\Oh(n \log n)}$ we can construct 
 words $x$ and $y$ such that $x = (0,m)$ and $y = (2^s-1,m)$ in $\BG$. 
 Now by \prref{rem:reddivBS} we have $m \mid s$ \IFF $x \simG y$. 
 The number of words of length $2^{\Oh(n \log n)}$ is at most 
 $2^{2^{\Oh(n \log n)}}$. 
 \qed\end{proof}

\section{Generic case analysis}\label{sec:generic}
 
 \newcommand{\PrO}[2]{\mathrm{Pr}_{0,n}^{#1}\hspace*{-0.1pt}\left[ \mathinner{#2}\right]}
 \newcommand{\Prk}[2]{\mathrm{Pr}_{#1}\hspace*{-0.1pt}\left[ \mathinner{#2}\right]}

\newcommand{\Prn}[2]{\mathrm{Pr}_{#1}\hspace*{-0.1pt}\left[ \mathinner{#2}\right]}
 \newcommand{\condProb}[3]
 {\mathrm{Pr}_{#1}\hspace*{-0.1pt}\left[ \mathinner{#2}\mathrel{\left|\vphantom{#2}\vphantom{#3}\right.} {#3}\right]}
 
 Let us  define a preorder between functions from $\N$ to $\R^{\geq 0}$ as follows. We let  
$f\preceq g$ if there exist $k \in \N$ and $\eps > 0$ such that
for almost all $n$ we have 
$$ f(n) \leq n^k g(n) + 2^{-\eps n}.$$
Moreover, we let $f\approx g$ if both, $f\preceq g$ and $g\preceq f$.
We are mainly interested in functions $f\approx 0$. These functions form an ideal in the ring of functions which are bounded by polynomial growth.
Moreover, if $f\approx 0$ then $g\approx 0$ for  $g(n) \in f(\theta(n))$.
The notion $f\approx g$ is therefore rather flexible and simplifies some formulae. 
We consider cyclically reduced words over $\Sig = \{a,\ov a, b, \ov b\}$ of length $n$ with uniform distribution. 
This yields a function $p(n) = \Prob{\exists y: x \simG y \wedge y\in H}$. 
We prove $p(n) \approx 0$. More precisely, we are interested in  the following result.

\begin{theorem}\label{thm:gen}
There is a strongly generic algorithm that decides in time $\Oh(n^4)$
on cyclically reduced input words  $x,y\in \{a,\ov a, b, \ov b\}^*$ with $\abs{xy} \in \theta(n)$ whether $x \simG y$. 
\end{theorem}
In the preceding section we have described the algorithm for the conjugacy problem. Hence, it remains to show that it runs strongly generically in $\Oh(n^4)$.
We give two proofs of \prref{thm:gen}. The first one is given in \prref{sec:dyckwordpair}. It
 uses a pairing by Dyck words. It is a little bit tedious, but self-contained and 
elementary. The second proof is given in \prref{sec:backbase}. 
It is based on a more general characterization which applies to all finitely generated  HNN extensions, see \prref{thm:backba}. To the best of our knowledge this characterization has not been stated elsewhere. The proof is not very hard, but in order to derive \prref{thm:gen} we need additional results from the literature. 
   
\subsection{Pairing with Dyck words: First proof of \prref{thm:gen}}\label{sec:dyckwordpair}


\newcommand{\inproofgen}{
 In order to switch {}from $\PrO{}{\cdots}$ to $\Prk{m}{\cdots}$ we consider the 
 \emph{block structure} $B(x)$ of a word $x \in \os{b, \ov b}\cdot \Sig^*$. 
 We define $B(x)$ as the tuple $(e_1,e'_1 \lds e_k,e'_k)$ 
 for  $x = \bet_1^{e_1} \alp_1^{e'_1} \cdots \bet_k^{e_k} \alp_k^{e'_k}$
 where $e_i, e'_i >0$, with the exception that possibly $e'_k =0$, 
 $\bet_i \in \os{b, \ov b}$, and  $\alp_i \in \os{a, \ov a}$.
 
 \noindent Let ${\wt E}_{k,m} = \set{(e_1,e'_1 \lds e_k,e'_k)}{\sum_{i=1}^k e_i = 2m 
\wedge \sum_{i=1}^k e'_i = n-2m}.$ For each 
$\wt e \in {\wt E}_{k,m}$ we obtain $\PrO{}{B(x) = \wt e\,} \in \theta(2^{2k}3^{-n})$ and $\Prk{m}{B(x) = \wt e\,} \in \theta(2^{2k}3^{-n})$. In particular, we have $\sum_{m =0}^{ \floor{n/4}} 
\sum_k \sum_{\wt e\in {\wt E}_{k,m}}\PrO{}{ B(x) = \wt e\,} \leq n 2^{n}3^{-n} \approx 0$ because 
$k \leq 2m$ for $\wt e \in {\wt E}_{k,m}$. 
Moreover, $\condProb0{x\in H}{B(x) = \wt e\,}= \condProb{m}{x\in H}{B(x) = \wt e\,}$. Indeed,
both values are equal to $2^{-2k}$ for  $e'_k >0$ and equal to $2^{1-2k}$ for $e'_k = 0$. This yields:
{\allowdisplaybreaks
\begin{alignat*}{3}
& \PrO{}{x \in H}\;  &\approx & \sum_{m = \ceil{n/4}}^{n} 
\sum_k \sum_{\wt e\in {\wt E}_{k,m}}\PrO{}{x \in H \wedge B(x) = \wt e\,}\\ 
&&\approx & \sum_{m = \ceil{n/4}}^{n} \sum_k \sum_{\wt e\in {\wt E}_{k,m}}\Prk{m}{x \in H \wedge B(x) = \wt e\,}\\ 
&& = & \sum_{m = \ceil{n/4}}^{n}\Prk{m}{x \in H \wedge \abs x = n } \leq \sum_{m = \ceil{n/4}}^{n}\Prk{m}{x \in H}\\
&& \approx\; &\Prk{\ceil{n/4}}{x \in H} \approx\; 0  \text{ by \prref{lem:maingen}}.
\end{alignat*}}
}

\newcommand{\proofmatch}{

\begin{proof}
Let $x$ be given  as its $\bet$-factorization $x = \bet_1 \alp_1\lds \bet_{2n} \alp_{2n}$. 
In order to compute $\scalb x d$, we scan $d= d_1 \cdots d_{2n}$ from left to right with $d_i \in \os{\lk \rk}$. 
We stop at each $j$ where $d_j = \rk$. Let $i$ be the corresponding 
index such that 
$d_i d_j$ is a matching pair in the Dyck word $d$. We have $i < j$. 
For fixed $j$, the \proba that $\bet_j = \ov {\bet_i}$ depends on $\bet_{j-1}$, only. We have  $\condProb{n}{\bet_{j} = \ov{\bet_i}}{{\bet_{j-1}} = {\bet_i}} = 1/3$ and $\condProb{n}{\bet_{j} = \ov{\bet_i}}{\ov{\bet_{j-1}} = {\bet_i}} = 2/3$.
Thus, $\Prk{n}{\bet_{j} = \ov {\bet_i}} \leq 2/3$. Moreover, for $j= i+1$ we obtain $\Prk{n}{\bet_{j} = \ov{\bet_i}} = 1/3$. Now, $\Prk{n}{\scalb x d = 1}$ implies in addition that 
for $j = i+1$ we must have $\bet_i = b$. In that case 
we calculate
$$\Prk{n}{\bet_{i} =b \wedge  \bet_{i+1}= \ov b}= \condProb{n}{\bet_{i+1} =\ov b }{\bet_{i}= b}\Prk{n}{\bet_{i} =b} \leq (1/3)\cdot(2/3). 
$$
The result follows. 
\qed
\end{proof}
}

\newcommand{\proofprobalpha}{
\begin{proof}
For real valued random variables $X$ we let $\Abs X = \sqrt{\sum_{k \in \Z} \Prob{X=k}^2}$.
Let us consider first an integer valued random variable $X$ which is given by some word
of the form $u\bet a^X \bet'v$. The distribution $\Prob{X=k}$ depends on $\bet, \bet'$, only. 
If $\bet= \bet'$ then $\Prob{X=k} = \frac{3^{-\abs k}}{2}$ for $k\in \Z$. 
If $\bet\neq \bet'$ then $\Prob{X=0} =0$ and  $\Prob{X=k} =3^{-\abs k}$ for $k\neq0$.
Thus, if $\bet= \bet'$ then ${\Abs X}^2 = 5/16$; and if $\ov \bet= \bet'$ then ${\Abs X}^2 = 1/4$. Hence:
\begin{equation}\label{eq:norm}
{\Abs X}^2 \leq 5/16.
\end{equation}
Next, consider a word of the form $u\bet a^X \bet'w \bet'' a^Y \ov \bet v$ with $\bet, \bet', \bet'' \in \os {b , \ov b}$ under the assumption that 
$\bet'w \bet'' = (r,m)$ in \BG where $(r,m) \in \sdZ = H$. The random variables $X$ and $Y$ are independent and define another random variable $Z$ (with values in $\Z[1/2]$) by the equation 
$(X,0)\cdot (r,m) \cdot (Y,0) = (Z,m)$ in $\BS12$, \ie $Z= X +r +2^m Y$. 
Hence, for $k \in \Z$ we obtain 
\begin{equation}\label{eq:Z}
\Prob{Z=k} = \sum_{i\in \Z}\Prob{X=i}\Prob{Y=2^{-m}(k-r-i)}.
\end{equation}
Note that $\Prob{Y=2^{-m}(k-r-i)}= 0$ unless $2^{-m}(k-r-i)\in \Z$. The numbers $m,k,r\in \Z$ are fixed 
and $2^{-m}(k-r-i) = 2^{-m}(k-r-j)$ implies $i=j$. Thus, we can define a new
random variable $Y'$ with the distribution $\Prob{Y'=i} = \Prob{Y=2^{-m}(k-r-i)}$. 
By \prref{eq:Z} and Cauchy-Schwarz inequality
$$\Prob{Z=k} = \sum_{i\in \Z}\Prob{X=i}\Prob{Y'=i} \leq \Abs{X}\Abs{Y'}.$$ 
Since $\Abs{Y'}\leq \Abs{Y}$,
we obtain $\Prob{Z=k} \leq \Abs{X}\Abs{Y}$. Finally,  by \prref{eq:norm}
\begin{align}\label{eq:Zk}
\Prob{Z=k}  \leq 5/16.
\end{align}

Now, let  $d = d_1 \cdots d_{2n}$ with $d_i \in B$ be a Dyck word and consider 
indices $i< j-1$ such that $(i,j)$ is a matching pair. 
(This means $d_id_j = \lk\rk$ and $d_{i+1} \cdots d_{j-1}$ is a non-empty Dyck word.) 
Let $n' = \frac{j-i+1}{2}$ and $d' = d_{i+1} \cdots d_{j-1}$. Next, we claim that 
\begin{align}\label{eq:aa}
\condProb{n'}{\scalp {x} {d_id'd_j}=1 } {\scalp {y} {d'}= 1 \wedge x= \ov b yb} \leq 5/16.
\end{align}
Note that \prref{eq:aa} refers to the measure $\mu_{n'}$ and thus, $x$ runs over those reduced words 
in $\Sig^*$ with $\abs{x}_\bet = 2n'$.
In order to see this inequality, consider a word $\ov b y b$ such that $\scalp {y} {d'}= 1$.
The word $y$ must contain two positions where letters from $\os{b, \ov b}$ appear because $j> i+1$.
Thus, we can write 
$y= \ov b a^X \bet w\bet' a^Y b$ such that $\bet w \bet' = (r,m)$ in $\BG$; and we can read
$X$ and $Y$ as integer valued random variables as before. For the derived random variable
$Z$ defined by $Z = X + r  + 2^m Y$ we obtain $\Prob{Z=0}  \leq 5/16$ by \prref{eq:Zk}. 
But $ \Prob{Z=0}$ is equal to $\condProb{n'}{\scalp {\bet y \tilde \bet} {d_id'd_j}=1 } {\scalp {y} {d'}= 1 \wedge \bet \tilde\bet = \ov b b}$. Hence, the claim. 

The other situation considers words of the form  $x =  b  y \ov b$. 
Again, we want to show
\begin{align}\label{eq:bovb}
\condProb{n'}{\scalp {x} {d_id'd_j}=1 } {\scalp {y} {d'}= 1 \wedge x= b y\ov b} \leq 5/16.
\end{align}
This is a more complicated situation and we need a case distinction about the structure of $d'= d_{i+1} \cdots d_{j-1}$.
We let $k$ denote the index which matches $i+1$ and $\ell$ matches the index $j-1$. 
For $\scalb {b y \ov b} {d_id'd_j}=1$, we can write $by\ov b  = ba^e \bet u \ov \bet y''\ov b$.  (Throughout we let  $\bet \in \os{b,\ov b}$ and $u,v,w,y\in \Sig^*$). 
But actually more is true. Assume $\bet = \ov b$ then index $i$ must match index $i+1$, but here we have 
$i+1 < j$, a contradiction. Hence, we conclude $\bet = b$. By symmetry, it follows that we can write $by\ov b = b a^e b  w \ov b a^f\ov b$.

\noindent{\bf Case $k > i +2$.} In this case we consider 
words $by\ov b $ which can be written as $by\ov b = b a^e b a^X\bet u \bet' a^Y\ov b v\ov b$ such that
 $\scalp{b a^X\bet u \bet' a^Y\ov b}{d_{i+1} \cdots d_k} = 1$. This implies
$\bet u \bet' = (r,0) \in \sdZ = H$ and $v = (s,q) \in  H$. Here, $X$ and $Y$ are random variables as above. In this setting,
 $\scalp{b y\ov b}{d_i d_j} = 1$ forces $Z= 0$ where $Z = X+r+ Y -q$. Inequality \prref{eq:Zk}  yields $\Prob{Z=0} \leq 5/16$. This shows \prref{eq:bovb}
 in the case $k > i +2$.

 \noindent{\bf Case $\ell < j-2$.} Symmetric to the precedent case. 
 
 \noindent{\bf Case $k =i +2$ and $\ell = j-2$.} We claim that this implies $k < \ell$.  Indeed, assume $\ell \leq k$ then we must have $i+1 = \ell$ and therefore 
 $i+1 = j-2$. Thus, $d' = d_{i+1} d_{i+2}$. But then 
 $\scalp{b y\ov b}{d_i\cdots d_{i+3}} = 1$ implies  $b y\ov b = b a^e b a^m \ov b a ^f \ov b$ with 
 $m \neq 0$, \ie $y=a^e t^ma^f \in H$ with $m\neq 0$. A contradiction because for $m \neq 0$ we have 
 $b y \ov b \notin H$ and $d$ is not successful. Thus, $i < k < \ell < j$. Now,  $\scalp{b y\ov b}{d_id'd_j} = 1$ implies  
$b y\ov b = b a^e b a^X \ov b \, u \, b a^Y \ov ba^f \ov b$.
Again, $X$ and $Y$ are random variables as above. Let $u = (r,m) \in \sdZ = H$.
 We have $b a^X \ov b = t^X $ and $b a^Y \ov b = t^Y$ in $\BG$. Thus, 
 $\scalp{b y\ov b}{d_id'd_j} = 1$ implies  $Z +m = 0$ where $Z = X+Y$. With the same arguments as in \prref{eq:Zk} we derive  $\Prob{Z=-m} \leq 5/16$. This shows \prref{eq:bovb}
 in the final case $k =i +2$ and $\ell = j-2$, too.

 Now, \prref{lem:probalpha} follows from \prref{eq:aa} and \prref{eq:bovb} since $n-k$ matching pairs $(i,j)$ exist in
 $d$ with $i+1 < j$. 
\qed\end{proof}
}

\begin{proof}
By \prref{thm:lea}, there is an algorithm deciding $x \simG y$ which runs in time $\Oh(n^4)$ for inputs which cannot be conjugated to  elements in $H$. Hence, we only have to bound the number of cyclically reduced  words of length $m \in \theta(n)$ which can be conjugated to some element in $H$. For simplicity of notation we assume $m=n$. 
A reduced word in $\Sig^n$ can be identified with a
random walk without backtracking in the Cayley graph of $\BG$ with generators $a$ and $b$. 
We encode reduced words over  $\Sig$ of length $n$ in a natural way as words in $\OO= \Sig\cdot\{1, 2, 3\}^{n-1}$. On $\OO$ we choose a uniform \proba (\eg, if the $i$-th letter is $b$ then the $i+1$-st letter is $a$, $ \ov a$, or $ b$ with equal probability $1/3$). 
Because we are interested in conjugacy, we compute the 
\proba under the condition that $x\in \OO$ is cyclically reduced. 
(Actually this does not change the results but makes the analysis smoother.) 
The \proba that $x\in \OO$ is cyclically reduced is  at least $2/3$ for all $n$. 
Let $C \sse \OO$  be the subset of cyclically reduced words. 
We show $\condProb{}{\exists y: x \simG y \wedge y\in H}{x \in C}\approx 0$. The question whether there exists some $y$ with $x \simG y$ is answered 
by calculating \Breduction{}s for a transposition of $x$. The set 
$C$ is closed under transpositions and it is no restriction to  assume that $\abs{x}_\bet \geq 1$. Therefore, we can choose the transposition that
$x' =vu$ where $x =uv$ such that the first letter of $x'$ is $\bet\in \os{b, \ov b}$. There are at most $n$ such transpositions. Hence,
\begin{alignat*}{3}
&\condProb{}{\exists y: x \simG y \wedge y\in H}{x \in C} &\; \approx\; &
\condProb{}{x\in H}{x \in C}\\
&=\;  \Prob{x\in H \wedge x \in C}\cdot \Prob{x \in C}^{-1}
&\;  \leq\; &  \Prob{x\in H}\cdot \Prob{x \in C}^{-1} \leq \frac{3}{2}\Prob{x\in H}.
\end{alignat*}
It is therefore enough to prove $\Prob{x\in H}\approx 0$. 
We switch the probability space and we embed $\OO$ into the space $\Sig^*$
with  a measure $\mu_{0,n}$ on $\Sig^*$ which concentrates on 
$\OO$, \ie $\mu_{0,n}(\OO) =1$. Within $\OO$ we still have a uniform distribution for $\mu_{0,n}$. In order to emphasize this change of view, we write $\Prob{\cdots } = \PrO{}{\cdots}$.
 We are now interested in  words $x \in \os{b, \ov b}\cdot\Sig^*$ which contain exactly 
 $2m$ letters $\bet\in \os{b, \ov b}$ for $m\geq 1$. (The number $\abs{x}_\bet$ must be even if $x \in H$.)
 Each such word can be written as a $\bet$-factorization of the 
 form $x = \bet_1 \alp_1 \ldots \bet_{2m} \alp_{2m}$ where 
 $\alp_i = a^{e_i}$ with $e_i \in \Z$. This defines a new measure $\mu_m$
 on $\Sig^*$ which is defined as follows. We start a random walk without backtracking with either $b$ or $\ov b$  with equal 
 \proba. 
 For the next letter there are always $3$ possibilities, each is chosen with \proba $1/3$. We continue as long as the random walk contains at most $2m$ letters from $\os{b, \ov b}$. This gives a corresponding 
 \proba on $\Sig^*$ which is concentrated on those words with $\abs{x}_\bet = 2m$. We denote the
 corresponding \proba 
 by $\Prk{m}{\cdots}$. 
\iflatin 
In the appendix it is shown that $\PrO{}{x \in H}\approx \sum_{m = \ceil{n/4}}^{n}\Prk{m}{x \in H}$. Hence, it remains to show $\Prk{m}{x \in H} \approx\; 0 $. \else \inproofgen Hence, the proof of \prref{thm:gen} is reduced to 
show \prref{lem:maingen}.\fi

{}From now on we work with the measure $\mu_n$ and the corresponding \proba $\Prk{n}{\cdots}$ for $n \geq 1$. 
Thus, we may assume that our \proba space contains only those words $x$ which have  $\bet$-factorizations of the 
 form $x = \bet_1 \alp_1 \ldots \bet_{2n} \alp_{2n}$
 with $\alp_i \in a^\Z$.
 \iflatin \else
The following result is the main lemma for the analysis of the generic case. \fi
\begin{lemma}\label{lem:maingen}
We have $\Prk{n}{x \in H} \leq (8/9)^n$. 
\end{lemma}

The proof of \prref{lem:maingen}  is based on a ``pairing'' with Dyck words:
Define a new alphabet $B = \os {\lk ,\rk}$ where $\lk$ is an opening 
left-bracket and $\rk$ is the corresponding  closing right-bracket.
The set of Dyck words $D_n$ is the set of words in $B^{2n}$ with correct bracketing. 
The number of Dyck words is well-understood, we have $\abs{D_n} = \frac{1}{n+1}\binom {2n}{n} \leq 4^n$. Thus,  $\abs{D_n} = C_n$, where
$C_n$ is the $n$-th Catalan number.
The connection between  Dyck words and \Breduction{}s is as follows.
\Breduction{}s are defined for words 
$\os{a,\ov a, t, \ov t, b, \ov b}^*$.  
Consider a $\bet$-factorization of the 
 form $x = \bet_1 \alp_1 \ldots \bet_{2n} \alp_{2n}$ with $\alp_i \in a ^\Z$.
 If  $x \in H$, then there exists a sequence of Britton reductions which transforms $x$ into $\wh x \in \os{a,\ov a, t,\ov t}^*$. We call such a sequence a \emph{successful \Breduction}.  Every successful \Breduction defines in a natural way a Dyck word by assigning an opening bracket to position $i$ and a closing bracket to position $j$ if $\bet_{i} u \bet_{j}$ is replaced by a Britton reduction. 
 Moreover, Britton reductions are confluent on $H$. 
 In particular, this means that for $x \in H$ we can start a successful  \Breduction
by replacing all factors $\bet_{i} a^e \bet_{i+1}$ with $\bet_i = b = \ov{ \bet_{i+1}}$ and $e \in \Z$ by $t^e$
where $1 \leq i < 2n$.
Thus, if such a successful \Breduction is described by $d$, then we may assume that 
$d_id_{i+1} = \lk \rk$ whenever $\bet_{i} a^e \bet_{i+1} = b a^e \ov b$.
Vice versa, if $d_{i}d_{i+1} = \lk \rk$, then we must have $\bet_i = b = \ov{ \bet_{i+1}}$, otherwise 
$d$ is no description of any \Breduction for $x$ at all. 
Note that for each $i$ with $d_i = \lk$ there is exactly one $j$ 
which matches $d_i$. The characterization of $j$ is that 
$d_{i+1} \cdots d_{j-1}$ is a Dyck word and $d_j = \rk$.
If $d$ describes a \Breduction for $x$ and $(i,j)$ is a matching pair for $d$
then $\bet_i  \ov{ \bet_{j}}= \bet\, \ov \bet $ for some $\bet \in \os{b, \ov b}$. 
We therefore say that $x$ and $d$ \emph{match}  if the following two conditions are 
satisfied: 
\begin{enumerate}
\item For all $1 \leq i < 2n$ we have $d_{i}d_{i+1} = \lk \rk \iff \bet_{i}\bet_{i+1} = b\, \ov b$. 
\item For all $1 \leq i < j \leq  2n$ where $d_{i}d_{j} = \lk \rk$ is a matching pair 
we have 
$\bet_{i}\bet_{j} = \bet\, \ov \bet $. 
\end{enumerate}
We define $\scalb xd = 1$ if $x$ and $d$ match and ${\scalb xd} = 0$ otherwise. 
We refine this pairing by defining $\scalp xd = 1$ if $\scalb xd = 1$ and $d$ describes a successful \Breduction proving 
$x \in H$. Otherwise we let $\scalp xd = 0$. 
Clearly, 
\begin{equation}\label{eq:wunderer}
\Prk{n}{x\in H} \leq \sum_{d \in D_{n}} \Prk{n}{\scalp xd = 1}.
\end{equation}
Since $\abs {D_n} \leq 4^n$, the  proof of \prref{lem:maingen} reduces  to show that for every $d\in D_n$ we have
\begin{equation}\label{eq:wunder}
\Prk{n}{\scalp xd = 1} \leq (2/9)^n.
\end{equation}

\begin{lemma}\label{lem:match}
Let $d \in D_{n}$ be a Dyck word and $k= \abs{\set{i}{d_{i}d_{i+1} = \lk \rk}}$. Then we have 
$\Prk{n}{\scalb x d = 1} \leq 
(2/3)^{n-k}(2/9)^k$.
\end{lemma}

\iflatin \else \proofmatch \fi

\begin{lemma}\label{lem:probalpha}
Let $d \in D_{n}$ be a Dyck word and $k= \abs{\set{i}{d_{i}d_{i+1} = \lk \rk}}$. Then we have
 $$ \condProb{k}{\scalp{x}{d} =1}{\scalb{x}{d} =1} \leq (5/16)^{n-k}.$$
\end{lemma}

 \iflatin \else \proofprobalpha \fi
 \prref{lem:match} and \prref{lem:probalpha} enable us to calculate $\Prk{n}{\scalp{x}{d} =1}$ as follows:
 \begin{align*}
\Prk{n}{\scalp{x}{d} =1}&= \condProb{k}{\scalp{x}{d} =1}{\scalb{x}{d} =1}\cdot \Prk{n}{\scalb{x}{d} =1}\\
& \leq (5/16)^{n-k} \cdot (2/3)^{n-k} (2/9)^k \leq (2/9)^n. 
\end{align*}
 This shows \prref{eq:wunder} and therefore \prref{lem:maingen} which in turn implies 
 \prref{thm:gen}.
 \qed\end{proof}

\subsection{Computer Experiments} 
 \begin{figure}[t]
 \vspace{-2mm}
\begin{center}
\begin{tikzpicture}[scale=.75]%
	\begin{axis}[ymode= log,width=\textwidth,  height=0.24\textheight ,xlabel={$n$}]
 		\addplot[densely dashed,every mark/.append style={solid,fill=black},mark=square*]  table[x=n,y=ratio] {experiments2.csv};
       \end{axis}
\end{tikzpicture} 
\vspace{-4mm}
\caption{Portion of reduced words $x\in H$ with $\abs{x}_\bet = 2n$,
sampling $11\cdot 10^9$ words.
}\label{fig:plot}\end{center}
\vspace{-5mm}
\end{figure}
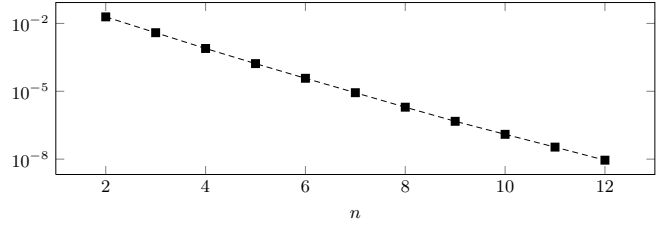

%
%
%
%
%
 We have conducted computer experiments with a sample of $11\cdot 10^9$ (\ie 11 billion)
 random words $x \in \Sig^*$ with $4 \leq \abs{x}_\bet = 2n \leq 24$, see  \prref{fig:plot}. Moreover, for $n=14$ our random process did not find a single $x\in H$.
 The experiments confirm $\Prk{n}{x \in H}\approx 0$. 
The initial values seem to 
 suggest $\Prk{n}{x \in H} \in \Oh(0{.}25^n)$. This is much better than the upper bound of \prref{lem:maingen}, but our 
 proof used very rough estimations in \prref{eq:wunderer} and \prref{eq:wunder}, only. Hence, a difference is no surprise.

\section{Back-to-base probability in HNN extensions: Second proof of \prettyref{thm:gen}}\label{sec:backbase}
This section has been added to the arXiv version in  November 2014, only. 
The motivation has been to give an alternative proof of \prettyref{thm:gen} which uses some known  results {}from literature.
For convenience of the reader there is some overlap with 
material in \prref{sec:dyckwordpair}. This allows an independent reading. 
In the following we investigate the general situation of an HNN extension $G$ which is given as $G = \Gen{H,b}{bab^{-1}=\phi(a), a \in A}$ with 
a finitely generated base group $H$. 
By the \emph{Back-to-base probability} we mean the probability 
that a random walk in the associated Cayley graph of $G$ ends 
in the base group $H$. 
In order to make the statement precise we fix the following 
notation. We let $H$ be the base group which is generated by some
finite subset $\Sig \sse H$  such that $\Sig = \Sig^{-1}$. 
We use a symmetric set of generators in order to apply \prettyref{prop:quasiiso}. (In fact, \prettyref{prop:quasiiso} is false for non-symmetric generating sets, in general.)
We let $A$ and $B$ be isomorphic subgroups of $H$ and  
$\phi: A\to B$ be a fixed isomorphism between them. Then, as usual, $G = \Gen{H,b}{bab^{-1}=\phi(a), a \in A}$  denotes the 
corresponding HNN extension of $H$ with stable letter $b$. 
By $\Del$ we denote the set $\Del= \Sig \cup \smallset{b, \ov b}$
where $\ov b = b^{-1}$. Thus, the
``evaluation of words over $\Del$'' defines a monoid presentation
$\eta: \Del^* \to G$, which is induced by the inclusion $\Del \sse G$. 
Recall that for $x \in \Del^*$ and 
$a \in \Del$ we denote by $\abs{x}_a$ the number of occurrences of the letter $a$ in the word $x$, and we let $\abs{x}_\bet = \abs{x}_b +\abs{x}_{\ov b} $. For $x \in \Del^*$ let $\wh x  \in \Del^*$ denote a \Breduced word such that $\eta(x) = \eta(\wh x)$ in $G$. Using this notation  
let us define $\Abs{x}_\bet$ by $\Abs{x}_\bet = \abs{\wh x}_\beta$.

For each $n \in \N$ we view $\Del^n$ as a \proba space with a uniform  distribution. Thus, we consider 
random walks in the Cayley graph of $G$ w.r.t.{} the generating set $\Del$ where each outgoing edge is chosen with equal probability. In contrast to 
\prref{sec:dyckwordpair} random walks may backtrack, i.e., they are not necessarily reduced words. 
We aim to show the following result. 
 
\begin{theorem}\label{thm:backba}
 Let $G = \Gen{H,b}{bab^{-1}=\phi(a), a \in A}$ be an HNN extension of $H$ and $\eta: \Del^* \to G$ as above.
 Then we have $A\neq H \neq B$ \IFF $\set{x\in \Del^*}{\eta (x)\not\in H}$  is strongly generic in $\Del^*$.
\end{theorem}
\begin{remark}
 In terms of amenability of Schreier graphs (see e.g., \cite{CeccheriniSilbersteinGH98,KMSS1}) we can restate \prettyref{thm:backba} as follows: Let $G = \Gen{H,b}{bab^{-1}=\phi(a), a \in A}$ be an HNN extension of $H$ and $\eta: \Del^* \to G$ as above.
 The Schreier graph $\GG(G,H,\Del)$ is non-amenable if and only if $A\neq H \neq B$.
\end{remark}

Before we prove \prref{thm:backba} let us show 
how to derive  \prettyref{thm:gen} from \prref{thm:backba}. 
We use the  following two propositions (see also \cite{KMSS1}).

\begin{proposition}[{\cite[Prop.{} 38, Thm.{} 51]{CeccheriniSilbersteinGH98}}]\label{prop:quasiiso}
  Let $G$ be a finitely generated group and $H \leq G$ be a subgroup. Let $\eta: \Del^* \to G$, $\eta': \Del'^* \to G$ two monoid presentations of $G$.  Then, $\Del^* \setminus \eta^{-1}(H)$ is strongly generic in $\Del^* $ \IFF  $\Del'^* \setminus \eta'^{-1}(H)$ is strongly generic in $\Del'^*$. 
 \end{proposition}

 \begin{proposition}[\cite{Bartholdi99,Cohen82,Grigorchuk77}]\label{prop:grigor}
  Let $G$ be a finitely generated group, $H \leq G$ be a subgroup, and $\eta: \Del^* \to G$ be a monoid presentation of $G$.  
  Let $\Xi$ be the set of reduced words of $\Del^*$.
  Then, $\Del^* \setminus \eta^{-1}(H)$ is a strongly generic
  in $\Del^*$ if and only if $\Xi \setminus \eta^{-1}(H)$ is strongly generic in $\Xi$.
 \end{proposition}

In order to see \prettyref{thm:gen} we proceed as follows: Let $\Xi$ denote the set of reduced words in  $\{a,\ov a, b, \ov b\}^*$ and  $\eta:\{a,\ov a, b, \ov b\}^* \to \BG$ the canonical presentation. Then  \prettyref{thm:backba}, \prettyref{prop:quasiiso}, and \prettyref{prop:grigor} show that $\Xi \setminus \eta^{-1}(H)$ is strongly generic in $\Xi$. 
 Now, with the same arguments as in \prref{sec:dyckwordpair} 
 it follows that elements which cannot be conjugated into $H$ form a strongly generic set in $\Xi$.

 Now, we turn to the proof of \prettyref{thm:backba}. It covers the rest of this section. 
  First, we consider $A = H =  B$. Then $G$ is a semidirect product 
 $G = H \rtimes \Z$. Let $\pi_2:G \to \Z$ the projection onto the second component. Then we have 
 $\eta(x) \in H$ \IFF $\pi_2(\eta(x)) = 0$.
 Since $\Del$ can be viewed as a constant, it is
not hard to see that we have  $\Prob{\eta(x) \in H} \in \Theta(1/\sqrt{ n})$. (Actually, if $\abs \Del$ is not viewed as a constant we obtain 
a more precise estimation. Since the expected value for ${\abs x}_\bet$ is $n/ 2\abs \Del$ one can show  $\Prob{\eta(x) \in H} \in \Theta(\sqrt{\abs\Del / n})$. But we do not need this for our purpose.)

 The second  case is  $A = H \neq B$. For example, $G$ is the Baumslag-Solitar group $\BS12$. We content ourselves with a 
 lower bound  on  $\Prob{\eta(x) \in H}$. 
 We begin with a the conditional \proba:
 \begin{equation}\label{eq:Ahb}
\condProb{}{\eta(x) \in H}{{\abs x}_\bet = 2m}\geq \frac{\binom {2m+1}{m}}{(m+1)2^m}
\in  \Theta(m^{-1.5}).
\end{equation}
To see this observe that, due to  $A = H$,  a \Breduction on a word $x \in \Del^*$ leads always to $H$ if both, ${\abs x}_b = {\abs x}_{\ov b}$ and for every prefix $y$ of $x$ we have ${\abs y}_b \geq  {\abs y}_{\ov b}$. Thus, $\eta(x) \in H$ as soon as 
the projection of $x$ onto $\smallset{b, \ov b }^*$ is a Dyck word. 
As we noticed earlier, the number of Dyck words of length $2m$ is the 
$m$-th Catalan number $\frac{1}{m+1}\binom {2m}{m} \in \Theta(m^{-1.5})$. 
We obtain 
a trivial estimation $\Prob{\eta(x) \in H} \in \OO(n^{-2.5})$ which is good enough because it means that for  $A= H$ the set $\set{x\in \Del^*}{\eta(x)\not\in H}$ is not strongly generic in $\Del^*$. However, using some standard Chernoff bounds and the fact that the expected value for ${\abs x}_\bet$ is $n/ 2\abs \Del$,  
we can state for $A = H$ a more precise  upper and  lower bound as follows:
\begin{equation}\label{eq:Ahb}
\Prob{\eta(x) \in H}\in \Oh(\sqrt{\abs \Del/ n}) \cap  \Omega((\abs \Del/  n)^{1.5}). 
\end{equation}
Finally, let us consider the most interesting case $A \neq H \neq B$.
This is the situation \eg in the Baumslag group $\BG$. 
In order to finish the proof of \prref{thm:backba} we have to show 
$\Prob{\eta(x) \in H}\approx 0$. This covers the rest of this section. 
As we have done in  \prettyref{sec:generic} we switch the probability space. We embed $\Del^n$ into the space $\Del^*$
with  a measure $\mu_{0,n}$ on $\Del^*$ which concentrates its mass on 
$\Del^n$ (\ie $\mu_{0,n}(\Del^n) =1$) with  corresponding \proba
 $ \PrO{}{\cdots}$. 
We now have to show that $\PrO{}{\eta(x)\in H}\approx 0$
 if $A\neq H \neq B$. Let  $\mu_m$ be the  measure
 on $\Del^*$ which is defined by reading letters from $\Del$ each with equal probability as long as the random walk contains at most $m$ letters $\bet \in \os{b, \ov b}$. This gives a corresponding 
 \proba on $\Del^*$ which is concentrated on those words with $\abs{x}_\bet = m$. We denote the
 corresponding \proba 
 by $\Prk{m}{\cdots}$. 
 Still there is a close connection between these probabilities. 
 In particular: 
  \begin{align}\label{eq:mamamia}
 & \PrO{}{\abs{x}_\beta = m\,} = \binom{n}{m} \cdot (2/\abs{\Del})^m \cdot (1 - 2/\abs{\Del})^{n-m} = \Prk{m}{\abs{x} = n\,} \\
 & \condProb{0,n}{\eta(x) \in H \,}{\abs{x}_\beta = m} = 
  \condProb{m}{\eta(x) \in H\,}{\abs{x} = n}
 \end{align}
 
 Since $\PrO{}{\abs{x}_\beta = m\,} \approx 0$ for $m \leq n/\abs{\Delta}$ we can perform a similar computation as in \prettyref{sec:dyckwordpair}:
 {\allowdisplaybreaks
 \begin{align*}
 \PrO{}{\eta(x) \in H}\;  & =  \quad \sum_{m = 0}^{n} \PrO{}{\eta(x) \in H \wedge\abs{x}_\beta = m\,}\\ 
&\approx  \sum_{m = \ceil{n/\abs{\Delta}}}^{n}  \PrO{}{\eta(x) \in H \wedge\abs{x}_\beta = m\,}\\ 
&=  \sum_{m = \ceil{n/\abs{\Delta}}}^{n}  
\condProb{0,n}{\eta(x) \,}{\abs{x}_\beta = m}
\cdot \PrO{}{\abs{x}_\beta = m\,}\\ 
&=  \sum_{m = \ceil{n/\abs{\Delta}}}^{n}  
\condProb{m}{\eta(x) \in H\,}{\abs{x} = n}\cdot \Prk{m}{\abs{x} = n\,}\\ 
&=  \sum_{m = \ceil{n/\abs{\Delta}}}^{n}  \Prk{m}{\eta(x) \in H \wedge\abs{x} = n\,}\\ 
  &\leq \sum_{m = \ceil{n/\abs{\Delta}}}^{n}\Prk{m}{\eta(x)  \in H}.\end{align*}}
 Therefore, it is enough to show that $\Prk{m}{\eta(x) \in H}\approx 0$ as a function in $m$. 
 
 There is also a natural probability distribution on $\Sig^*$ which is formally defined by $\mu_{0}$ (N.B.{} $\mu_{0}$ is different from $\mu_{0,n}$!) Indeed, we have $\mu_0(\Sig^*) = 1$ and 
 the distribution on $\Sig^*$ is given by a random walk which stops with \proba $2/\abs \Del$ and, if it does not stop, then it chooses the next letter with equal probability. In order to emphasize that the mass of 
 $\mu_{0}$ is on $\Sig^*$ we also write $\Prk{\Sig}{y } = \Prk{0}{y}$ for 
 $y \in \Sig^*$. 

\begin{lemma}\label{lem:indexprob}
  For all $\gamma \in \Sig^*$ and $\beta \in \smallset{b, \ov b}$ we have
 $$\Prk{\Sig}{\eta(\beta \gamma y \ov \beta) \not \in H} \geq \frac{2}{\abs{\Delta}^2}.$$
\end{lemma}

 \begin{proof}
By symmetry we may assume $\bet= b$. We have to show that $\Prk{\Sig}{\eta(\gamma y) \notin A}\geq 2/{\abs \Del}^2$. 
 We consider the cases $\eta(\gam) \notin A$ and $\eta(\gam)  \in A$
 separately. 
 For  $\eta(\gam) \not \in A$ we obtain 
 $$\Prk{\Sig}{\eta(\gamma y) \notin A} \geq \Prk{\Sig}{ y = 1 } = 2/\abs \Del \geq  2/{\abs \Del}^2. $$
 For  $\eta(\gam)  \in A$ and $a \in \Sig$ we obtain 
 $\eta(\gam a) \in A$ \IFF $\eta(a) \notin A$. Since 
 $A\neq H$ and $\Sig$ generates $H$, there must be  some letter $a\in \Sigma$ with $\eta(a)\not \in A$. 
Therefore, in the second case  
 $$\Prk{\Sig}{\eta(\gamma y) \notin A} \geq \Prk{\Sig}{ y = a} = 2/{\abs \Del}^2. $$
\qed
 \end{proof}

 As before a $\beta$-factorization of $x\in \Del^*$ with $\abs{x}_\beta = m$ 
is written as a word
$x=\gam_{0}\bet_1\gam_{1}\ldots \bet_m\gam_{m}$  such that $\bet_i \in \os{b, \ov b}$ and $\gam_i \in \Sig^*$ for $1 \leq i \leq m$. 
Using the notion of $\beta$-factorization we define for all 
$0 \leq \ell \leq m$ 
a random variable $X_\ell:\Del^* \to \N$ as follows. We let 
$X_\ell(x) = \Abs{\gam_{0}\bet_1\gam_{1}\ldots \bet_\ell\gam_{\ell}}_\bet$. Another way to explain $X_\ell(x)$ is as follows. Choose any prefix 
$z$ of $x$ such that $\abs{z}_\beta = \ell$, compute the \Breduction 
$\wh z$ of $z$ and let  $X_\ell(x) = \abs{\wh z}_\beta$, \ie 
$X_\ell(x) = \Abs{z}_\beta$. The differences $Y_i= X_i - X_{i-1}$ 
define random variables $Y_i$ for $1 \leq i \leq m$ with values in  
$\smallset{-1, 1}$. Clearly, 
 $X_\ell = \sum_{i=1}^{\ell} Y_i$ for all $0 \leq \ell \leq m$.
 Note that $X_0 = 0$ and $X_1 = Y_1= 1$ are constant functions. 
  
  Consider a $\beta$-factorization $x=\gam_{0}\bet_1\gam_{1}\ldots \bet_m\gam_{m}$ for $x$ with $\abs x_\beta= m$. For $1 \leq i \leq m$ let 
  $z_{i-1}$
  be  \Breduced such that $\eta(z_{i-1})= \eta(\gam_{0}\bet_1 \ldots \gam_{i-2}\bet_{i-1})$. 
  Then the $\beta$-factorization of $z_{i-1}$ becomes
  $\gam'_{0}\bet'_1\gam'_{1}\ldots \bet'_j\gam'_{j}$ for some 
  $j \leq i-1$. 
  Note that the last factor $\gam'_{j}$ can be, a priori, any word in $\Sig^*$.  Now, it depends only on  the factors $\bet'_j\gam'_{j}$ and 
  $\gam_{i-1}\bet_i$ whether or not the $\beta$-length of the Britton-reduced word 
  increases or decreases when reading the next factor $\gam_{i-1}\bet_i$.
  The \proba for that is described by the random variable $Y_i$. 
   For all $\eps \in \{-1,1\}^{i-1}$ \prettyref{lem:indexprob} shows
  \begin{equation}\label{eq:lisa}
\condProb{}{Y_i =1}{Y_j = \eps_j \text{ for } j<i} \geq 1/2 + 1/2 \cdot 2/\abs{\Del}^2 = 1/2 + 1/\abs{\Del}^2.
\end{equation}
Let $\set{Z_i}{i=1,\dots, m}$ be a set of $m$ independent random variables taking values in $\smallset{-1, 1}$ such that $\Prob{Z_i =1} = 1/2 +1/\abs{\Del}^2$ for $1 \leq i \leq m$. By \prref{eq:lisa} it follows that for every $\eps= (\eps_j)\in \smallset{-1,1}^{k-1}$ and $1 \leq k \leq m$ we have
\begin{equation}\label{eq:lisach}
\condProb{m}{Y_k = -1}{Y_j = \eps_j \;\forall j<k}\leq \Prob{Z_k = -1}.
\end{equation}
This observation is crucial in the proof of the next lemma.

\begin{lemma}\label{lem:Xm}
We have 
$$\Prk{m}{X_m =0} \leq \left(1 -  \frac{4}{\abs\Del^4}\right)^{m/2}.$$
\end{lemma}

\begin{proof}
 The assertion is trivial for $m=0$ or $m$ odd. Hence, let 
 $m \geq 2$ be even. First, let us show that 
 for all $p \in \Z$, $1 \leq k \leq \ell \leq m$, and $\eps = (\eps_j) \in \{-1,1\}^{k-1}$ we have 
 \begin{equation}\label{eq:nalja}
\condProb{m}{\sum_{i=k}^{\ell} Y_i \leq p}{Y_j = \eps_j \;\forall j<k} \leq \Prk{m}{\sum_{i=k}^{\ell} Z_i \leq p}.
\end{equation}
We prove \prref{eq:nalja} by induction on $k-\ell$. The case $\ell= k$ is trivial, hence let $\ell < k$.
\begin{align*}
 &\condProb{m}{\sum_{i=k}^{\ell} Y_i \leq p}{Y_j = \eps_j \;\forall j<k} \\
 & \quad=  \sum_{\eps_k = \pm 1}\condProb{m}{Y_k = \eps_k}{Y_j = \eps_j \;\forall j<k}\cdot \condProb{m}{\sum_{i=k+1}^{\ell} Y_i \leq p - \eps_k}{Y_j = \eps_j \;\forall j\leq k} \\
		      & \quad \leq\sum_{\eps_k = \pm 1}\condProb{m}{Y_k = \eps_k}{Y_j = \eps_j \;\forall j<k} \cdot \Prk{m}{\sum_{i=k+1}^{\ell} Z_i \leq  p - \eps_k}\\  
		      & \quad\leq  \sum_{\eps_k = \pm 1}\Prk{m}{Z_k = \eps_k} \cdot \Prk{m}{ \sum_{i=k+1}^{\ell} Z_i \leq  p - \eps_k} =  \Prk{m}{\sum_{i=k}^{\ell} Z_i \leq p}.
\end{align*}
We have to explain the inequality leading to the last line above. 
By \prref{eq:lisach} there is some $\del_{\eps,k} \geq 0$ such that 
$\condProb{m}{Y_k = -1}{Y_j = \eps_j \;\forall j<k} + \del_{\eps,k}= \Prk{m}{Z_k = -1}$. Thus, by definition, $\condProb{m}{Y_k = 1}{Y_j = \eps_j \;\forall j<k} - \del_{\eps,k}= \Prk{m}{Z_k = 1}$. Hence, the inequality follows from
$\Prk{m}{ \sum_{i=k+1}^{\ell} Z_i \leq  p - 1} \leq \Prk{m}{ \sum_{i=k+1}^{\ell} Z_i \leq  p + 1}$. 

 As a special case for $k=1$ and $\ell = m$ we obtain
\begin{align}\label{eq:watnu}
\Prk{m}{X_m \leq p}=  \Prk{m}{\sum_{i=1}^{m} Y_i \leq p}
  \leq \Prk{m}{ \sum_{i=1}^{m} Z_i \leq  p }.
		      \end{align}		     
In order to prove the lemma it is enough to consider $p=0$. We get
 \begin{align*}
  \Prk{m}{X_m = 0}  
  &\leq \Prob{\sum_{i=1}^{m} Z_i \leq 0} 
  = 
  \sum_{\eps = (\eps_j)\in \oneset{-1,1}^{m} \atop 
  \abs { \set {j}{\eps_j = 1}}\leq m/2
  }
  \prod_{i=1}^{m}\Prob{Z_i =\eps_i}\\
 		  &\leq 2^m \cdot \left(\frac{1}{2} - \frac{1}{{\abs \Del}^2}\right)^{m/2}\cdot \left(\frac{1}{2} + \frac{1}{{\abs \Del}^2}\right)^{m/2} =  \left(1 -  \frac{4}{\abs\Del^4}\right)^{m/2}.
 \end{align*}
 \qed
 \end{proof}
Hence, we have concluded the proof of \prettyref{thm:backba}
 because \prettyref{lem:Xm} implies in particular $\Prk{m}{X_m = 0} \approx 0$.

\section*{Conclusion} We have investigated the complexity of the \CP in two important groups 
in combinatorial group theory. The \CP in $\BS12$ is $\TC$-complete. If  division in \PC{}s is 
 non-elementary in the worst case, then the \CP in $\BG$ is non-elementary on the average, but
 solvable in $\Oh(n^4)$ on a strongly generic subset. This is a striking contrast underlying
 the importance of generic case complexity on natural examples.  
In order to derive the result about generic case complexity, we proved a more general result about HNN extensions.  We showed that  $G = \Gen{H,b}{bab^{-1}=\phi(a), a \in A}$ has a non-amenable Schreier graph with respect to the base group $H$ \IFF $A\neq H \neq B$.

%
\end{document}